\documentclass[a4paper, english]{lipics-v2016}
\usepackage{microtype}%if unwanted, comment out or use option "draft"
\usepackage[utf8]{inputenc}
\usepackage{cite}

\usepackage{amsmath,amsfonts,amssymb,color}
\usepackage{amsmath,amsfonts,amssymb,color}
\usepackage[ruled,vlined,linesnumbered]{algorithm2e}
\usepackage{wrapfig}

\newcommand{\rr}{\mathbb{R}}

\renewcommand{\O}{\mathcal{O}}

\newcommand{\ep}{\varepsilon}
\renewcommand{\vec}[1]{\ensuremath{{\bf #1}}}

\newcommand{\la}{\leftarrow}

\newcommand{\ceil}[1]{\left\lceil #1 \right\rceil}

\renewcommand{\bold}[1]{\textbf{#1}}
\newcommand{\parabold}[1]{\smallskip
\noindent\bold{#1}}

\newcommand{\comm}[1]{~~~~~~\text{(#1)}}

\DeclareMathOperator*{\argmin}{arg\,min}
\DeclareMathOperator*{\argmax}{arg\,max}

\theoremstyle{plain}
\newtheorem{claim}[theorem]{Claim}
\newtheorem{proposition}[theorem]{Proposition}

\newcommand{\calB}{\mathcal{B}}

\newcommand{\calF}{\mathcal{F}}

\newcommand{\calL}{\mathcal{L}}
\newcommand{\calM}{\mathcal{M}}

\newcommand{\vecb}{\vec{b}}

\newcommand{\vece}{\vec{e}}
\newcommand{\vecf}{\vec{f}}

\newcommand{\vecp}{\vec{p}}

\newcommand{\vecs}{\vec{s}}

\newcommand{\vecu}{\vec{u}}

\newcommand{\vecw}{\vec{w}}
\newcommand{\vecx}{\vec{x}}

\newcommand{\vecz}{\vec{z}}
\newcommand{\vecll}{\ensuremath{\boldsymbol\ell}}

\newcommand{\bbp}{\mathbf{p}}

\newcommand{\bbx}{\mathbf{x}}
\newcommand{\bby}{\mathbf{y}}

\newcommand{\bbB}{\mathbf{B}}

\newcommand{\ps}{p^*}
\newcommand{\Phims}{\Phi_{\textsf{MS}}}
\newcommand{\Phimsi}{\Phi_{\textsf{\emph{MS}}}}
\newcommand{\Phicpf}{\Phi_{\textsf{CPF}}}
\renewcommand{\mod}[1]{\left| #1 \right|}
\newcommand{\norm}[1]{\left\lVert#1\right\rVert}
\newcommand{\xs}[1]{x^{*, #1}}

\newcommand{\rint}{\mathsf{rint}}
\newcommand{\inner}[2]{\left\langle ~#1 ~,~ #2~\right\rangle}
\newcommand{\KL}{\texttt{KL}}
\newcommand{\KLe}{\emph{\texttt{KL}}}

\newcommand{\vecps}{\vec{\ps}}

\title{Tracing Equilibrium in Dynamic Markets via Distributed
  Adaptation}% Tatonnement and Proportional Response Dynamic}
\titlerunning{Tracing Equilibrium in Dynamic Markets via Distributed
  Adaptation} %optional, in case that the title is too long; the running title should fit into the top page column

\author[1]{Yun Kuen Cheung}
\author[2]{Martin Hoefer}
\author[3]{Paresh Nakhe}
\affil[1]{Max-Planck-Institute for Informatics, Saarland Informatics Campus, Germany\\
  \texttt{ycheung@mpi-inf.mpg.de}}
\affil[2]{Goethe University Frankfurt/Main, Germany\\
  \texttt{mhoefer@cs.uni-frankfurt.de}}
\affil[3]{Goethe University Frankfurt/Main, Germany\\
	\texttt{nakhe@em.uni-frankfurt.de}}
\authorrunning{Y.\,K. Cheung and M. Hoefer and P. Nakhe}
\Copyright{Yun Kuen Cheung and Martin Hoefer and Paresh Nakhe}%mandatory, please use full first names. LIPIcs license is "CC-BY";  http://creativecommons.org/licenses/by/3.0/

% mandatory: Please choose ACM 1998 classifications from http://www.acm.org/about/class/ccs98-html . E.g., cite as "F.1.1 Models of Computation". 
% mandatory: Please provide 1-5 keywords
\begin{document}

\maketitle

% !TEX root = main.tex

\begin{abstract}
Competitive equilibrium is a central concept in economics with numerous applications beyond markets, such as scheduling, fair allocation of goods, or bandwidth distribution in networks. Computation of competitive equilibria has received a significant amount of interest in algorithmic game theory, mainly for the prominent case of Fisher markets. Natural and decentralized processes like tatonnement and proportional response dynamics (PRD) converge quickly towards equilibrium in large classes of Fisher markets. Almost all of the literature assumes that the market is a \emph{static environment} and that the parameters of agents and goods do not change over time. In contrast, many large real-world markets are subject to frequent and dynamic changes.

In this paper, we provide the first provable performance guarantees of discrete-time tatonnement and PRD in markets that are subject to perturbation over time. We analyze the prominent class of Fisher markets with CES utilities and quantify the impact of changes in supplies of goods, budgets of agents, and utility functions of agents on the convergence of tatonnement to market equilibrium. Since the equilibrium becomes a dynamic object and will rarely be reached, our analysis provides bounds expressing the distance to equilibrium that will be maintained via tatonnement and PRD updates. Our results indicate that in many cases, tatonnement and PRD follow the equilibrium rather closely and quickly recover conditions of approximate market clearing. Our approach can be generalized to analyzing a general class of Lyapunov dynamical systems with changing system parameters, which might be of independent interest. 
\end{abstract}

% !TeX root = main.tex

\section{Introduction}

A central concept to understand trade in large markets is the notion of competitive or market equilibrium. The computational aspects of competitive equilibria have been a central theme in algorithmic game theory over the last decade, mainly for the prominent class of Fisher markets. In a Fisher market, there are a set of agents or buyers and a set of divisible goods. Each agent brings a budget of money to the market and wants to buy goods, for which she has an increasing and concave utility function. An equilibrium consists of a vector of prices and an allocation of goods and money such that (1) every buyer purchases the most preferred bundle of goods that she can afford, and (2) market clears (supply equals demand). 

There are successful approaches based on distributed adaptation processes for converging to market equilibria. For example, \emph{tatonnement} is governed by the natural intuition that prices for over-demanded goods increase, while under-demanded goods become cheaper. It provides an explanation how decentralized price adjustment can lead a market into an equilibrium state, thereby providing additional justification for the concept. Recently, several works derived a detailed analysis and proved fast convergence of discrete-time tatonnement in markets~\cite{CodenottiMV2005,CF2008,CCR2012,CCD2013,CheungC2014,CheungC2016}.

It is well-known that network rate control is closely related to Fisher market equilibria~\cite{Kelly1997,KMT1998,KV2002}. Towards this end, distributed market dynamics called \emph{proportional response dynamics} (PRD) were proposed and analyzed in the context of peer-to-peer network sharing~\cite{WuZ2007,LLSB08}. These dynamics avoid the usage of prices and work directly on the exchange and allocation of goods. PRD and its generalizations converge toward market equilibria in the full range of CES Fisher markets~\cite{Zhang2011,BDX2011,CCT2018}.

While tatonnement and PRD rely on dynamic change of prices and allocation, the existing literature assumes that the market and its properties (agents, budgets, utilities, supplies of goods) remain static and unchanged over time. In fact, to the best of our knowledge, all of the existing work on computation of market equilibrium in algorithmic game theory assumes that the market is essentially a \emph{static environment}. In contrast, in many (if not all) applications of markets, the market itself is subject to dynamic change, in the sense that supplies of goods changes over time, agents have different budgets at their disposal that they can spend, or the preferences of agents expressed via utility functions evolve over time. Analyzing and quantifying the impact of dynamic change in markets is critical to understand the robustness of market equilibrium in general, and of price adaptation dynamics like tatonnement in particular. %\MH{More motivation and discussion about novelty/impact: Timely innovation? Upcoming, drastic changes in many economies world-wide? Or rather something with internet advertising?}

In this paper, we initiate the algorithmic study of \emph{dynamic markets} in the form of dynamically evolving environments. Our interest lies in the performance of dynamic adaptation processes like tatonnement. We analyze a discrete-time process, where in each round $t$ tatonnement provides a price for each good, which is then updated using the excess demand for each good. In each round $t$ the excess demand comes from a possibly different, adversarially perturbed market. This dynamic nature of markets gives rise to a number of interesting issues. Notably, even when in each round $t$ the market has a unique equilibrium, over time this equilibrium becomes a dynamic object. As such, exact market equilibria can rarely or never be reached. Instead, we consider how tatonnement can trace the equilibrium by maintaining a small distance (in terms of suitably defined notions of distance), which also results in approximate clearing conditions. For PRD, we apply a similar approach based on adaptation of the allocation of goods.

More formally, we consider the prominent class of Fisher markets with agent utilities that exhibit constant elasticity of substitution (CES). In this versatile framework, we analyze the impact of changes in supply of goods, budgets of agents, and their utility parameters. The adaptation approaches equilibrium conditions, but since the equilibrium is moving, prices and allocations follow and chase the equilibrium point over time. Our analysis provides distance bounds, which can be seen as a quantification of the extent of out-of-equilibrium trade due to the interplay of market perturbation and adaptation of agents.

Technically, the majority of our analysis is concerned with quantifying the impact of perturbation in market parameters on several potential functions that guarantee convergence of the dynamics. The results then follow by a combination with the convergence guarantees for static markets. In fact, this approach constitutes a powerful framework to analyze a large variety of protocols and dynamics that are well-understood in static systems, when these systems become subject to dynamic perturbation. In the Appendix~\ref{sect:applications}, we discuss further examples -- gradient descent for strongly convex functions and diffusion for load balancing in networked systems -- where we quantify the performance of natural dynamics in the presence of system perturbations.

\parabold{Contribution and Outline.} After presenting necessary preliminaries in Section~\ref{sec:prelim}, we describe in Section~\ref{sect:dynamic-market-general-model} the general model for dynamic Fisher markets with CES utilities and a general convergence result. In the subsequent sections, we discuss the insightful case of CES markets with gross-substitute condition. In these markets, the total misspending (absolute excess demand times price) over all goods is a natural parameter to quantify the violation of market clearing conditions. Moreover, one round of tatonnement updates in static markets is known to decrease misspending by a multiplicative factor~\cite{CCR2012}. In Sections~\ref{sect:dynamic-supply} and~\ref{sect:dynamic-utility}, we consider markets where the supply of goods, the budgets of agents, and the utility function of the agents are subject to dynamic perturbation, respectively. We quantify the impact of perturbation on the misspending in the market. These bounds reveal that the change is often a rather mild additive change in misspending. Together with the fact that tatonnement decreases the misspending multiplicatively, we see that the price adaptation is indeed able to incorporate and adapt to the changes quickly. Overall, the dynamics can trace the equilibrium point up to a distance that evolves from the change in a small number of recent rounds.

We can provide similar results for a more general approach for CES markets based on a convex potential function~\cite{CCD2013}. A slight disadvantage is that this potential function does not have an equally intuitive interpretation as the misspending function. On the other hand, it applies to tatonnement in all CES Fisher markets (even without gross-substitute property). The discussion of these results is deferred to Appendix~\ref{app:CPF}.

The technique we apply for markets can be executed much more generally for a class of dynamical systems, which we outline in Section~\ref{sect:general-framework}. These systems have a set of control parameters (e.g., prices in markets, or strategic decisions in games) and system parameters (e.g., supplies or utilities in markets, or payoff values in games). Moreover, these systems admit a Lyapunov function, and a round-based adaptation process for the control parameters (e.g., tatonnement in markets, or best-response dynamics in classes of games) that multiplicatively decreases the Lyapunov function in a single round. Our results provide a bound on the value of the Lyapunov function when the system parameters are subject to dynamic change. We discuss two additional examples of such systems based on minimization of strongly convex functions and network load balancing in Appendix~\ref{sect:applications}. 

In Section \ref{sect:damped-PR} we use a further generalization of the technique based on Bregman divergence to show that proportional response dynamics can successfully trace equilibrium in gross-substitute CES Fisher markets. A more general framework of dynamical systems governed by progress in Bregman divergence is discussed in Appendix~\ref{app:Bregman}. It seems likely that a similar analysis based on our techniques can be conducted for many more sophisticated systems with significantly more complex dynamics.

\parabold{Related Work.} Competitive equilibrium and tatonnement date back to Walras~\cite{walras} in 1874. The existence of equilibrium was established in a non-constructive way for a general setting by Arrow and Debreu~\cite{ArrowDebreu1954} in 1954. Computation of equilibrium has been a central subject in general equilibrium theory. In the past 15 years, there has been impressive progress on devising efficent algorithms for computing equilibria, e.g., using network-flow algorithms~\cite{DPSV2008,Orlin2010,DuanM2015,DuanGM16,BeiGH16,BeiGHM16}, the
ellipsoid method~\cite{Jain2007} or the interior point method~\cite{Ye2008}.

Decentralized adaptation processes such as tatonnement are important due to their simple nature and plausible applicability in real markets. Tatonnement is broadly defined as a process that increases (resp.~decreases) the price of a good if the demand for the good is more (resp.~less) than the supply. The price updates are \emph{distributed}, since the price adjustment for each good is in the direction of its own excess demand and is independent of the demands for other goods.

Arrow, Block and Hurwitz~\cite{ABH1959} showed that a continuous version of tatonnement converges to an equilibrium for markets satisfying the weak gross substitutes (WGS) property. The recent algorithmic advances provide new insights in analyzing tatonnement~\cite{CodenottiMV2005,CCD2013}. Cole and Fleischer~\cite{CF2008} proposed the \emph{ongoing market model}, in which \emph{warehouses} are introduced to allow out-of-equilibrium trade, and prices are updated in tatonnement-style \emph{asynchronously}, to provide an \emph{in-market} process which might capture how real markets work. There has been significant recent interest in further aspects of ongoing markets or asynchronous tatonnement~\cite{CFR2010,CCR2012,CheungC2014,CheungC2016}.

In contrast, proportional response dynamics are a class of distributed algorithms that originated in the literature on network bandwidth sharing. These dynamics work without prices and come with convergence guarantees in classes of static network exchange, where goods have a uniform value~\cite{WuZ2007}. For the special class of Fisher markets with linear utilities these dynamics can be cast as a form of gradient descent~\cite{Zhang2011,BDX2011,CCT2018}.

Notions of games and markets with perturbation and dynamic change are only very recently starting to receive increased interest in algorithmic game theory. For example, recent work has started to quantify the average performance of simple auctions and regret-learning agents in combinatorial auctions with dynamic buyer population~\cite{LykourisST16,FosterLLST16}. In these scenarios, however, equilibria are probabilistic objects and convergence in the static case can only be shown in terms of regret on average in hindsight. Moreover, the main goal is to bound the price of anarchy.

% !TeX root = main.tex

\renewcommand{\vecb}{\mathbf{b}}
\renewcommand{\vecp}{\mathbf{p}}
\renewcommand{\vecu}{\mathbf{u}}
\renewcommand{\vecs}{\mathbf{s}}
\renewcommand{\vecw}{\mathbf{w}}
\renewcommand{\vecx}{\mathbf{x}}
\renewcommand{\vecz}{\mathbf{z}}

%\newcommand{\calM}{\mathcal{M}}

%\newcommand{\AvgPhims}{\overline{\Phi}_{\textsf{MS}}}
%\newcommand{\AvgPhicpf}{\overline{\Phi}_{\textsf{CPF}}}

%\vspace{-0.25cm}

\section{Preliminaries}
\label{sec:prelim}

\parabold{Fisher Markets.}
In a Fisher market, there are $n$ \emph{goods} and $m$ \emph{buyers}. Each buyer $i$ has an amount $b_i$ of \emph{budget}. Buyer $i$ has a \emph{utility function} $u_i$ representing her preference. For bundles $\vecx_i^1 = (x_{ij}^1)_{j=1,\ldots,n}$ and $\vecx_i^2 = (x_{ij}^2)_{j=1,\ldots,n}$, if $u_i(x^1_{i1},\cdots,x^1_{in}) > u_i(x^2_{i1},\cdots,x^2_{in})$, then she prefers $\vecx_i^1$ to $\vecx_i^2$. We denote the vector of budgets by $\vecb =(b_i)_{i=1,\ldots,m}$ and the vector of utility functions by $\vecu = (u_i)_{i=1,\ldots,m}$. Let $B = \sum_i b_i$ be the total budget in the market.

Given a vector $\vecp = (p_j)_{j=1,\ldots,n}$ of (per-unit) \emph{prices} for each good, buyer $i$ requests a \emph{demand bundle} of goods that maximizes her utility function subject to the budget constraint:
$\hat{\vecx}_i = \arg \max \left\{ u_i(\vecx_i) :\nonscript\; \sum_{j=1}^n x_{ij}\cdot p_j~\le~b_i \right\}$.
In general, the $\argmax$ is a set of bundles. In this paper, we concern strictly concave utility function only, for which there is a \emph{unique} demand bundle.
%Let $\hat{u}_i(\vecp) = u_i(\vecx_i^*)$ denote the objective value of the demand bundle.

The sum of amount of good $j$ purchased by all buyers is the \emph{demand for good $j$}, denoted by $x_j = \sum_{i=1}^m \hat{x}_{ij}$. The \emph{supply of good $j$} is $w_j$, and we set $\vecw = (w_j)_{j=1,\ldots,n}$. Let $\vecz =(z_j)_{j=1,\ldots,n}$ be the vector of \emph{excess demand}, i.e., demand minus supply: $z_j = x_j - w_j$.
%\YKC{Since $\circ^*$ is usually used for denoting quantities at equilibrium, I changed $\vecx_i^*$, the demand of buyer $i$, to $\hat{vecx}_i$.}

A pair $(\vecx^*,\vecp^*)$ is a \emph{competitive} or \emph{market equilibrium} if
(1) each vector $\vecx^*_i$ is a demand of agent $i$ at prices $\vecp^*$,
(2) for each good $j$ with $p_j^* > 0$, demand is equal to supply (i.e., $p_j^* \cdot z_j = 0$), and 
(3) for each good $j$ with $p_j^* = 0$, demand is at most supply (i.e., $z_j\le 0$).
An equilibrium price vector $\vecp^*$ is also called a vector of \emph{market clearing prices}.

%\YKC{ Table (for internal use):\\
%There are $n$ goods and $m$ buyers. We use index $i$ to run over buyers, while we use $j,k$ to run over goods.\\
%Each buyer $i$ has utility function $u_i$, in which the coefficients are $a_{ij}$, for goods $1\le j\le n$.\\
%The budget of buyer $i$ is $b_i$.\\
%Price of good $j$ is $p_j$; equilibrium price is $\ps_j$.\\
%Demand of buyer $i$ for good $j$ is denoted by $x_{ij}$, and the total demand for good $j$ is $x_j$.\\
%Supply of good $j$ is $w_j$. The excess demand for good $j$ is $z_j := x_j - w_j$.
%}

\parabold{CES Utility Functions.}
A prominent class of utility functions in markets are utility functions with Constant Elasticity of Substitution (CES). They have the form
$u_i(\vecx_i) = \left( \sum_{j=1}^n a_{ij} \cdot (x_{ij})^{\rho} \right)^{1/\rho}$, where $1\ge \rho > -\infty$ and all $a_{ij} \ge 0$.

For $\rho<1$ and $\rho\neq 0$, buyer $i$'s demand for good $j$ is 
$$\hat{x}_{ij} ~=~ b_i \cdot \frac{(a_{ij})^{1-c} (p_j)^{c-1}}{\sum_{k=1}^n (a_{ik})^{1-c} (p_k)^c},~~\text{where}~~c~=~ \frac{\rho}{\rho-1}\enspace.$$ 
%and the attained utility is 
%$$\hat{u_i}(\vecp) = \frac{b_i}{\left(\sum_{k=1}^n (a_{ik})^{1-c} (p_k)^c\right)^{1/c}}\enspace.$$

%\parabold{Clearing Prices.} Define market equilibrium.

\parabold{Dynamic Markets.}
For CES Fisher markets, tatonnement is known to converge quickly to equilibrium under static market conditions. We here consider a dynamic market where in the beginning of each round $t = 1,\ldots,T$ our tatonnement dynamics propose a vector of prices $\vecp^t$. Dynamic market parameters like budgets $\vecb^t$, supplies $\vecw^t$ and utility functions $\vecu^t$ are manifested, which can be different from their value in previous rounds $0,\ldots,t-1$. Agents request a demand bundle based on the prices $\vecp^t$ and market $\calM^t = (\vecu^t, \vecb^t, \vecw^t)$, which yields a vector of excess demands $\vecz^t$. Then the system proceeds to the next round $t+1$.

We first provide a basic insight that lies at the core of the analysis and manages to lift convergence results for a class of static markets to a bound for dynamic markets from that class. Formally, assume that the following properties hold:
\begin{description}
\item[Potential:] There is a non-negative potential function $\Phi(\calM, \vecp)$, for every market $\calM = (\vecu, \vecb, \vecw)$ and every price vector $\vecp$. It holds $\Phi(\calM, \vecp) = 0$ if and only if $\vecp$ is a vector of clearing prices for market $\calM$.
\item[Price-Improvement:] The price dynamics satisfy $\Phi(\calM, \vecp^t) \le (1-\delta) \cdot \Phi(\calM, \vecp^{t-1})$, for some $1 \ge \delta > 0$ and every market $\calM$.
\item[Market-Perturbation:] The market dynamics satisfy $\Phi(\calM^t, \vecp) \le \Phi(\calM^{t-1}, \vecp) + \Delta^t$, for some values $\Delta^t \ge 0$ and every price vector $\vecp$.
\end{description}

\begin{proposition}\label{prop:marketConverge}
  Suppose the price and market dynamics satisfy the Potential, Price-Improvement, and Market-Perturbation properties. Then 
  \[
    \Phi(\calM^T,\vecp^T) \quad \le \quad (1-\delta)^T\cdot\Phi(\calM^0,\vecp^0) \; + \; \sum_{t=1}^{T-1} (1-\delta)^{T-t} \Delta^t\enspace.
  \]
  Let $\Delta = \max_{t=1,\ldots,T} \Delta^t$, then it follows for any $t \le T$
  \[
    \Phi(\calM^T,\vecp^T) \quad \le \quad \sum_{\tau=t+1}^{T} (1-\delta)^{T-\tau} \Delta^\tau \; + \; \frac{(1-\delta)^{T-t}}{\delta} \cdot \Delta \; + \; (1-\delta)^T\cdot\Phi(\calM^0,\vecp^0) \enspace.
  \]
\end{proposition}

The proof follows by a direct application of the three properties. We prove it for a much more general class of dynamic systems with Lyapunov functions in Section~\ref{sect:general-framework}.

Consider the three terms in the latter bound for $\Phi$. The first term captures the impact of \emph{recent} changes to the market. The second term bounds the effect of all \emph{older} changes. The third term decays exponentially over time. Hence, when the process runs long enough, the potential is only affected by \emph{recent changes} of the market, while all older changes can be accumulated into a constant term based on $\Delta$ and $\delta$. Intuitively, the price dynamics follows the evolution of the equilibrium up to a ``distance'' of $\Delta/\delta$ in the potential function value. Hence, if market perturbation $\Delta$ is small and price improvement $\delta$ is large, the process succeeds to maintain market clearing conditions almost exactly.

\section{Dynamic Fisher Markets and Misspending}
\label{sect:dynamic-market-general-model}

%\subsection{Tatonnement and Potential}
For simplicity, we here describe our techniques for CES markets $\calM$ with gross-substitutes property, i.e., when all buyers have CES utilities with $1 > \rho > 0$. For a study of general CES Fisher markets see Appendix~\ref{app:CPF} below. 

The tatonnement process we analyze here updates prices in each round based on the excess demand in the last round, i.e.,
\begin{equation}\label{eq:multi-tat-update}
p_j^t ~\leftarrow~ p_j^{t-1} \cdot \left[~1 + \lambda \cdot \min\left( \frac{x_j^{t-1} - w_j}{w_j}~,~ 1 \right)~\right]\enspace,
\end{equation}
where $\lambda < 1$ is a parameter depending on $\rho$. The \emph{misspending potential function}~\cite{CFR2010,CCR2012} is
$$
\Phims(\calM, \vecp) ~=~ \sum_{j=1}^n p_j \cdot \left|z_j\right|.
$$
The tatonnement process is known to have the Price-Improvement property based on the misspending potential function $\Phims$ in CES markets with $1 > \rho > 0$. More formally, let $\rho_{\max} := \max_{i=1}^m \rho_i$, if $\lambda \le \Theta(1-\rho_{\max})$, then
there exists $1 \ge \delta = \delta(\lambda) > 0$ such that $\Phims(\calM,\vecp^t) \le (1-\delta) \cdot \Phims(\calM,\vecp^{t-1})$.~\cite{CFR2010}
\newcommand{\vep}{\boldsymbol{\varepsilon}}
\newcommand{\vept}{\boldsymbol{\varepsilon}^t}
\newcommand{\veptau}{\boldsymbol{\varepsilon}^\tau}
\newcommand{\supply}{\text{\textsf{supply}}}
\renewcommand{\O}{\mathcal{O}}
\newcommand{\tp}{\tilde{p}\hspace*{0.01in}}
\newcommand{\tps}{\tp^{*}}
\newcommand{\tpst}{\tp^{*,t}}
\newcommand{\gmax}{\gamma_{\max}}
\newcommand{\gmin}{\gamma_{\min}}
\newcommand{\vecss}{\vec{s}\hspace*{0.01in}'}

\newcommand{\Phidms}{\Phi_{\textsf{dMS}}}

\subsection{Dynamic Supply and Budgets}
\label{sect:dynamic-supply}

\parabold{Dynamic Supply.}
Let us first analyze the impact of changing supply on tatonnement dynamics and market clearing conditions. We normalize the initial supply $w_j^1 = 1$ for each good $j$. Suppose that the supplies are then changed additively\footnote{We here study additive change for mathematical convenience. The bounds can be adjusted to hold accordingly for multiplicative change.} by $\vept = (\ep_1^t,\ep_2^t,\cdots,\ep_n^t)$ at time $t$. We parametrize our bounds using the maximum supply change $\kappa = \max_{t} \|\vept\|$. 

\smallskip

\noindent\textbf{Assumption 1.}~Every price is universally bounded by some time-independent constant $P$, i.e., for any $j$ and any time $t$, we have $p_j^t \le P$.

\smallskip

Assumption 1 is made for technical reasons, but it is simple to satisfy by constant parameters of the market. For example, if all initial prices are at most $B$, then since $\lambda < 1$ Assumption 1 holds with $P = 2B$. The main result in this section is as follows.

\begin{proposition}
  For any $t \le T$,
  $$
    \Phimsi(\calM^T,\vecp^T) ~\le~ P \cdot \left(\sum_{\tau=t+1}^{T} (1-\delta)^{T-\tau} \|\veptau\|_1 ~+~ \frac{ (1-\delta)^{T-t}}{\delta}\cdot \kappa \right) + (1-\delta)^T \cdot \Phimsi(\calM^0,\vecp^0)\enspace.
  $$
\end{proposition}

\begin{proof}
Consider the misspending potential $\Phims$. Tatonnement satisfies the Price-Improvement property. Hence, to show the result, we establish the Market-Perturbation property. 

Note that the misspending potential can be given by
$$
\Phims(\calM^t,\vecp^t) ~=~ \sum_{j=1}^n p_j^t \cdot \left| x_j^t - 1 - \sum_{\tau=1}^t \ep_j^\tau \right|\enspace.
$$
Hence, by the triangle inequality and Assumption 1,
\begin{align*}
  \Phims(\calM^t,\vecp) ~%&\le~ \sum_{j=1}^n p_j \cdot \left| x_j^t - 1 - \sum_{\tau=1}^{t-1} \ep_j^\tau \right| ~~+~~ \sum_{j=1}^n p_j \cdot \left| \ep_j^t \right| \\
  &=~ \Phims(\calM^{t-1},\vecp) ~+~ \sum_{j=1}^n p_j \cdot \left| \ep_j^t \right| %\\
  ~~~\le~~~ \Phims(\calM^{t-1},\vecp) ~+~ P\cdot \|\vept\|_1\enspace.
\end{align*}
Thus, using Proposition~\ref{prop:marketConverge} with $\Delta^t = P\cdot \|\vept\|_1$ and $\Delta = P \kappa$, the proof follows.
%it holds for any $t \le T$
%$$
%\Phims(\calM^T,\vecp^T) ~\le~  P \cdot \left(\sum_{\tau=t+1}^{T} (1-\delta)^{T-\tau} \|\veptau\|_1 ~+~ \frac{ (1-\delta)^{T-t}}{\delta}\cdot n\kappa \right) + (1-\delta)^T \cdot \Phims(\calM^0,\vecp^0)\enspace.\qedhere
%$$
\end{proof}

%\YKC{Should we discuss (briefly) about the ``running away'' issue here? Here is some suggested text.}

\noindent\textbf{Remark.}$\;$
If the supplies of all goods shrink \emph{multiplicatively} by the
same factor of $(1-\beta)$ per time step, then in markets with CES
utility functions, it is well-known that the equilibrium price of
every good increases by a factor of $(1-\beta)^{-1}$ per time 
step. However, the tatonnement update rule allows the current price to
be increased by a factor of at most $(1+\lambda)$ per time step. Thus,
for plausible \emph{tracing} of equilibrium, $\lambda$ should satisfy $(1+\lambda) > (1-\beta)^{-1}$.

\parabold{Dynamic Budgets.}
%\label{sect:dynamic-budget}
We now analyze the impact of changing buyer budgets on tatonnement dynamics and market clearing conditions. Starting from the initial budgets, the budgets are then changed additively
%\footnote{We here study additive change for mathematical convenience. The bounds can be adjusted to hold accordingly for multiplicative change.} 
by $\vept = (\ep_1^t,\ep_2^t,\cdots,\ep_m^t)$ at time $t$. We
parametrize our bounds using the maximum budget change $\eta =
\max_{t} \|\vept\|_1$. For a proof of the following proposition, see Appendix~\ref{app:misspendingDynamicBudgets}.

\begin{proposition}
\label{prop:misspendingDynamicBudgets}
  For any $t \le T$,
  $$
    \Phimsi(\calM^T,\vecp^T) ~\le~ \sum_{\tau=t+1}^{T} (1-\delta)^{T-\tau} \|\veptau\|_1 ~+~ \frac{(1-\delta)^{T-t}}{\delta}\cdot \eta ~+~ (1-\delta)^T \cdot \Phimsi(\calM^0,\vecp^0)\enspace.
  $$
\end{proposition}

% !TEX root = main.tex

\subsection{Dynamic Buyer Utility}\label{sect:dynamic-utility}

In this section, we analyze the impact of changing the parameters $a_{ij}$ in the CES utility functions on tatonnement dynamics and market clearing conditions. Starting from the initial utility values, each $a_{ij}$ can in each round be changed by some multiplicative factor $\gamma_{ij}^t$. Let $\gamma^t = \max_{i,j} ((\gamma_{ij}^t)^{\frac{1}{1-\rho}},(1/\gamma_{ij}^t)^{\frac{1}{1-\rho}})$ and $\gamma = \max_t \gamma^t$. 

\begin{proposition}
  For any $t \le T$,
  \begin{align*}
    \Phimsi(\calM^T,\vecp^T) ~\le&~ B \cdot \left( \sum_{\tau=t+1}^{T}
      (1-\delta)^{T-\tau} \cdot \frac{2(\gamma^\tau - 1)}{\gamma^\tau
        + 1} ~+~ \frac{(1-\delta)^{T-t}}{\delta} \cdot \frac{2(\gamma
        - 1)}{\gamma + 1} \right) \\
&~+~ (1-\delta)^T \cdot \Phimsi(\calM^0,\vecp^0)\enspace.
\end{align*}
\end{proposition}

\begin{proof}
To show the result, we establish the Market-Perturbation property. Note that the misspending potential can be given by
\begin{align*}
\Phims(\calM^t,\vecp^t) ~=~ \sum_{j=1}^n p_j^t \cdot \mod{x_j^t - w_j}
                        &=~ \sum_{j=1}^n p_j^t \cdot \left| \sum_{i=1}^m b_i \cdot \frac{(a_{ij} \prod_{\tau=1}^t\gamma_{ij}^\tau)^{1-c} (p_j)^{c-1}}{\sum_{k=1}^n (a_{ik}\prod_{\tau=1}^t\gamma_{ik}^\tau)^{1-c} (p_k)^c} - w_j \right|\enspace.
\end{align*}

Using $a_{ij}^{t-1} = a_{ij} \prod_{\tau=1}^{t-1}\gamma_{ij}^\tau$ we derive
\begin{align*}
  \Delta^t ~&=~ \Phims(\calM^t,\vecp) - \Phims(\calM^{t-1},\vecp)\\
           ~&=~ \sum_{j=1}^n p_j \cdot \left(\left| \sum_{i=1}^m b_i \cdot \frac{(a_{ij}^{t-1}\gamma_{ij}^t)^{1-c} (p_j)^{c-1}}{\sum_{k=1}^n (a_{ik}^{t-1}\gamma_{ik}^t)^{1-c} (p_k)^c} - w_j \right| ~-~ \left| \sum_{i=1}^m b_i \cdot \frac{(a_{ij}^{t-1})^{1-c} (p_j)^{c-1}}{\sum_{k=1}^n (a_{ik}^{t-1})^{1-c} (p_k)^c} - w_j \right| \right)\\
           ~&\le~ \sum_{j=1}^n p_j \cdot \sum_{i=1}^m b_i \cdot \mod{ \frac{(a_{ij}^{t-1}\gamma_{ij}^t)^{1-c} (p_j)^{c-1}}{\sum_{k=1}^n (a_{ik}^{t-1}\gamma_{ik}^t)^{1-c} (p_k)^c} ~-~ \frac{(a_{ij}^{t-1})^{1-c} (p_j)^{c-1}}{\sum_{k=1}^n (a_{ik}^{t-1})^{1-c} (p_k)^c} }\\
           ~&=~ \sum_{i=1}^m b_i \cdot \sum_{j=1}^n \mod{ \frac{(a_{ij}^{t-1}\gamma_{ij}^t)^{1-c}\, p_j^c }{\sum_{k=1}^n (a_{ik}^{t-1}\gamma_{ik}^t)^{1-c}\, p_k^c} ~-~ \frac{(a_{ij}^{t-1})^{1-c}\, p_j^c }{\sum_{k=1}^n (a_{ik}^{t-1})^{1-c}\, p_k^c} } \enspace.
\end{align*}
For the rest of the proof, we construct an upper bound on the difference of two fractions. Fix a buyer $i$. We set $\alpha_j = \frac{(a_{ij}^{t-1})^{1-c}\, p_j^c}{\sum_{k=1}^n (a_{ik}^{t-1})^{1-c}\, p_k^c}$ and $\beta_j = (\gamma_{ij}^t)^{1-c}$. Moreover, we use $\mu = \gamma^t$ and observe $\mu \ge \beta_j \ge 1/\mu$. We let
\begin{align*}
\Delta s_j ~&=~ \mod{ \frac{(a_{ij}^{t-1}\gamma_{ij}^t)^{1-c}\, p_j^c }{\sum_{k=1}^n (a_{ik}^{t-1}\gamma_{ik}^t)^{1-c}\, p_k^c} ~
-~ \frac{(a_{ij}^{t-1})^{1-c} \, p_j^c }{\sum_{k=1}^n (a_{ik}^{t-1})^{1-c} \, p_k^c} }
~=~ \mod{ \frac{\alpha_j \beta_j}{\sum_k \alpha_k \beta_k} - \alpha_j } \enspace. %\\
%&=~ \mod{ \frac{\alpha_j \beta_j}{\sum_k \alpha_k \beta_k} - \frac{\alpha_j}{\sum_k \alpha_k} } \enspace.
\end{align*}
\begin{lemma}
\label{lem:misspendingDynamicUtility}
There exists a vector $(\beta'_1,\ldots,\beta'_m)$ with 
$$
\beta'_j = \begin{cases}
\mu & \text{ if } \frac{\alpha_j \beta'_j}{\sum_k \alpha_k \beta'_k} \ge \alpha_j \\
1/\mu & \text{ otherwise.}
\end{cases}
$$
such that 
$$
\sum_{j} \Delta s_j ~\le~ \sum_j \mod{ \frac{\alpha_j \beta'_j}{\sum_k \alpha_k \beta'_k} - \alpha_j }\enspace.
$$
\end{lemma}
For a proof of the lemma, see Appendix~\ref{app:misspendingDynamicUtility}. Now let $\beta'$ be a vector as in the previous lemma, let $S = \{j : \beta'_j = \mu\}$ and $R = G \setminus S$. Using $\alpha_S = \sum_{j \in S} \alpha_j$, we obtain
\begin{align*}
\sum_j \Delta s_j ~&\le~ \left( \sum_{j \in S} \frac{\alpha_j
    \mu}{\sum_{k \in S} \alpha_k \mu + \sum_{i \in R} \alpha_i / \mu}
  - \sum_{j \in S} \alpha_j \right) \\ 
&\hspace{0.5cm} ~+~  \left( \sum_{j \in R} \alpha_j - \sum_{j \in R} \frac{\alpha_j/\mu}{\sum_{k \in S} \alpha_k \mu + \sum_{i \in R} \alpha_i /\mu}  \right) \\
~&=~ \left( \frac{\mu \alpha_S}{\mu \alpha_S + (1 - \alpha_S) / \mu} - \alpha_S \right) ~+~  \left( 1 - \alpha_S - \frac{(1-\alpha_S)/\mu}{\mu \alpha_S + (1-\alpha_S) /\mu} \right) \\
~&=~ 1 - 2\alpha_S + \frac{ \left(\mu + \frac{1}{\mu}\right) \alpha_S - \frac{1}{\mu}}{\left(\mu-\frac{1}{\mu}\right) \alpha_S + \frac{1}{\mu}} \enspace.
\end{align*}
This expression is maximized at $\alpha_S = \frac{1}{\mu + 1}$ and yields an upper bound of $\sum_j \Delta s_j \le \frac{2(\mu - 1)}{\mu + 1}$. Thus, using Proposition~\ref{prop:marketConverge} with $\Delta^t \leq B \cdot \frac{2(\gamma^t - 1)}{\gamma^t + 1}$ and $\Delta \le B \cdot \frac{2(\gamma - 1)}{\gamma + 1}$, we are done.
%it holds for any $t \le T$
%%
%\begin{align*}
%\Phims(\calM^T,\vecp^T) ~&\le~ B \cdot \left( \sum_{\tau=t+1}^{T} (1-\delta)^{T-\tau} \cdot \frac{2(\gamma^\tau - 1)}{\gamma^\tau + 1} ~+~ \frac{(1-\delta)^{T-t}}{\delta} \cdot \frac{2(\gamma - 1)}{\gamma + 1}  \right)\\
%&\quad+~ (1-\delta)^T \cdot \Phims(\calM^0,\vecp^0)\enspace.\qedhere
%\end{align*}
\end{proof}

\section{Parametrized Lyapunov Dynamical Systems}\label{sect:general-framework}

In this section, we prove a general theorem, which includes as special case the bound shown for markets in Proposition~\ref{prop:marketConverge}. Our focus here are dynamical systems, in which time is discrete and represented by non-negative integers. Note, however, that the formulation below can be easily generalized to settings with continous time.

We assume that the dynamical system can be described by two sets of parameters. There is a set of \emph{control variables} that can be adjusted by an algorithm or a protocol. In addition, there is a set of \emph{system parameters} that can change in each round in an adversarial way. For example, in our analysis of markets in the previous section, the control variables are prices of goods, whereas system parameters can be supplies of goods, budgets of agents, or utility parameters. As another example, in games the control variables could be the strategy choices of agents, whereas system parameters are utility and payoff values of states. More generally, control variables could also be bird headings in a bird flock, while system parameters are wind direction or velocity, etc.

The classical theory of dynamical systems often studies the behaviour of systems with static system parameters. However, dynamical systems with varying system parameters often arise in practice (see Appendix \ref{sect:applications} for some examples). Here, we propose a simple framework to analyze Lyapunov dynamical systems with varying system parameters. More formally, the dynamical system $L$ is described by an initial \emph{control variable vector} $\vecp^0 \in \rr^n$ and an \emph{evolution rule} $F:\rr^n \rightarrow \rr^n$, which specifies how the control variables are adjusted. For each time $t\ge 1$, we have $\vecp^t ~=~ F(\vecp^{t-1})$.

The system $L$ is called a \emph{Lyapunov dynamical system} (LDS) if it admits a Lyapunov function $G:\rr^n\rightarrow \rr^+$ such that
\begin{enumerate}
	\item[(a)] for every fixed point (equilibrium) $\vecp$ of $F$ with $F(\vecp) = \vecp$ it holds $G(\vecp) = 0$;
	\item[(b)] for every $\vecp\in \rr^n$ it holds $G(F(\vecp)) ~\le~ G(\vecp)$.
\end{enumerate}
An LDS $L$ is called \emph{linearly converging} (LCLDS) if it further satisfies
\begin{enumerate}
	\item[(c)] there exists a \emph{decay parameter} $\delta = \delta(L) > 0$ such that for any $\vecp\in \rr^n$,
	$G(F(\vecp)) ~\le~ (1-\delta) \cdot G(\vecp)$.
\end{enumerate}

Let $\calL$ be a family of dynamical systems, while each dynamical system $L_{\vecs}\in \calL$ is parametrized by a \emph{system parameter} vector $\vecs\in \rr^d$. The family $\calL$ is called a family of \emph{parametrized, linearly converging LDS} (PLCLDS) if each $L_{\vecs} \in \calL$ is an LCLDS and $\delta(\calL) = \inf_{L_{\vecs}\in \calL} \delta (L_{\vecs}) > 0$. For each $L_{\vecs}$, we denote its evolution rule by $F_{\vecs}$ and its Lyapunov function by $G_{\vecs}$.

In many scenarios, particularly in agent-based dynamical systems, the control variables $\vecp$ change by the evolution rule that expresses, e.g., the sequential decisions of the agents, but the system parameters $\vecs$ can change in an exogenous (or even adversarial) way. However, in many cases the impact of changes in a single time step is rather mild. The following theorem states our \emph{recovery result} by relating the Lyapunov value to the magnitude of changes in each step. Intuitively, it characterizes the ``distance'' that the evolution rule maintains to a fixed point over the course of the dynamics.

\begin{theorem}\label{thm:meta}
Let $\calL$ be a PLCLDS with $\delta\equiv \delta(\calL) > 0$, let
$\vecs^0,\vecs^1,\ldots,\vecs^T$ denote the system parameter vectors
at times $0,1,\cdots,T$, respectively, and let $\Phi(\vecs^t,\vecp^t)
= G_{\vecs^t}(\vecp^t)$. Suppose that for every $t=1,\ldots,T$ the
system parameters $\vecs^{t-1},\vecs^t \in\rr^d$ invoke a change such
that for every $\vecp\in\rr^n$, we have
$
\Phi(\vecs^t, \vecp) ~\le~  \Phi(\vecs^{t-1},\vecp) + \Delta^t.
$
The initial control variable vector is denoted by $\vecp^0$, and the system evolves such that for every $t\ge 1$ we have $\vecp^t = F_{\vecs_{t-1}}(\vecp^{t-1})$. Then
$$
\Phi(\vecs^T,\vecp^T) ~\le~ (1-\delta)^T \cdot \Phi(\vecs^0,\vecp^0) ~+~ \sum_{t=1}^T (1-\delta)^{T-t} \cdot \Delta^t\enspace.
$$
Let $\Delta = \max_{t=1,\ldots,T} \Delta^t$, then it follows for any $t \le T$
\[
  \Phi(\vecs^T,\vecp^T) \quad \le \quad \sum_{\tau=t+1}^{T} (1-\delta)^{T-\tau} \Delta^\tau \; + \; \frac{(1-\delta)^{T-t}}{\delta} \cdot \Delta \; + \; (1-\delta)^T\cdot\Phi(\vecs^0,\vecp^0) \enspace.
\]
\end{theorem}
\begin{proof}
For any time $t\ge 1$,
\begin{align*}
\Phi(\vecs^t,\vecp^t) &~=~ G_{\vecs_t}(\vecp^t) ~\le~ G_{\vecs_{t-1}}(\vecp^t) + \Delta^t \\
&~=~ G_{\vecs_{t-1}}(F_{\vecs_{t-1}}(\vecp^{t-1})) + \Delta^t\\
&~\le~ (1-\delta) \cdot G_{\vecs_{t-1}}(\vecp^{t-1}) + \Delta^t ~~~=~ (1-\delta) \cdot \Phi(\vecs^{t-1},\vecp^{t-1}) + \Delta^t.
\end{align*}
Iterating the above recurrence yields the first result. For the second result, note that 
\begin{align*}
\sum_{\tau=1}^{t} (1-&\delta)^{T-\tau} \Delta^\tau ~\le~ \Delta
(1-\delta)^T \cdot \sum_{\tau=1}^{t} \left(\frac{1}{1-\delta}\right)^{\tau}\\ ~&=~ \Delta \cdot \frac{(1-\delta)^{T+1}}{\delta} \cdot \left(\left(\frac{1}{1-\delta}\right)^{t+1} - \frac{1}{1-\delta}\right)
~<~ \Delta \cdot \frac{(1-\delta)^T}{\delta} \cdot \left(\frac{1}{1-\delta}\right)^t.\qedhere
\end{align*}
\end{proof}
In the scenarios where $\sum_{t=1}^T \Delta^t = O(T^\alpha)$ for small constant $\alpha$,
we have the following corollary.
\begin{corollary}\label{corr:meta}
In the setting of Theorem~\ref{thm:meta}, if $\sum_{t=1}^T \Delta^t ~=~ \O(T^\alpha)$ for some constant $\alpha > 0$, then for any constant $\beta >0$,
$$
\Phi(\vecs^T,\vecp^T) ~\le~ \sum_{\tau=T-\ceil{\frac{\alpha+\beta}{\delta}\log T} + 1}^T \Delta^{\tau} ~+~ \O(T^{-\beta}) ~+~ (1-\delta)^T \cdot \Phi(\vecs^0,\vecp^0).
$$
\end{corollary}

As $T\rightarrow\infty$, the last two terms of the above inequality diminish.
The bound is dominated by the first term, which describes the impact of the changes in the \emph{recent} $\O\left( \frac{\log T}{\delta} \right)$ steps.

% !TEX root = main.tex

\section{Proportional Response Dynamics}
\label{sect:damped-PR}

%Damped Proportional Response Dynamic (DPRD), proposed by Cheung, Cole and Tao~\cite{CCT2018}, is a generalization of Proportional Response Dynamic (PRD).
In the Fisher market setting, the general protocol of proportional response dynamics (PRD) is as follows. In each round, each buyer $i$ splits her budget $b_i$ among the $n$ goods according to some rule, and send the bids to the sellers of the corresponding goods. Based on the bids gathered from all buyers, the seller of each good $j$ send back (simple) signals to buyers, which are then used by buyers for updating their bids in the next round. We summarize the notation and results we need from Cheung, Cole and Tao~\cite{CCT2018} below. When buyer $i$ splits her budget $b_i$ among the $n$ goods, let $b_{ij}$ denote the spending by her on good $j$. Let $\bbB$ denote $\{b_{ij}\}_{i\in [m],j\in [n]}$. Let $p_j := \sum_i b_{ij}$.

Consider the substitute domain, i.e., when the $\rho_i$ parameters of all buyers are strictly between $0$ and $1$. In each round, the seller of good $j$ distributes the good among buyers in proportion to the bids received, and then after receiving the goods, each buyer splits her budget in proportion to the utility generated from the quantity of each good received. More formally, let $p_j^t = \sum_i b_{ij}^t$, then the update rule is
\begin{equation}
\label{eqn:subst-CES-PR-rule}
b_{ij}^{t+1} ~=~ b_i \cdot a_{ij} \left(\frac {b_{ij}^{t}}{p_j^t}\right)^{\rho_i} ~\left/ \left(\sum_k a_{ik} \left(\frac {b_{ik}^{t}}{p_k^t}\right)^{\rho_i}\right)\right.
\end{equation}
%
%\PN{I think we should not discuss about complementary domain at all.} In the complementary domain, i.e., when the $\rho_i$ parameters of all buyers are strictly between $-\infty$ and $0$,
%in each round, the seller of good $j$ broadcasts the sum of all bids for good $j$, i.e., $p_j^t$,
%then each buyer split her budget according to the following rule:
%$$
%b_{ij}^{t+1} ~=~ b_i \cdot  \Big(\frac{a_{ij}} { {(p_j^t)}^{\rho_i}}\Big)^\frac{1}{1 - \rho_i} \left/ \left( \sum_k   \Big(\frac{a_{ik}}{ {(p_k^t)}^{\rho_i}}\Big)^\frac{1}{1 - \rho_i} \right) \right.
%$$

The \emph{Kullback-Leibler (KL) divergence} is similar to a distance measure. For vectors $\bbx$ and $\bby$ such that $\sum_j x_j ~=~ \sum_j y_j$, the explicit formula is $\KL(\bbx || \bby) ~:=~ \sum_j x_j \cdot \ln \frac{x_j}{y_j}$. The above update rule is equivalent to mirror descent w.r.t.~the KL divergence (but with different step sizes for different buyers) of the same function:
\begin{equation}\label{eq:PF-PRD}
g(\bbB) ~=~ - \sum_{ij} \frac{b_{ij}}{\rho_i} \log \frac{a_{ij} (b_{ij})^{\rho_i - 1} } { \left( p_j \right)^{\rho_i}  }\enspace,
\end{equation}
defined on the domain
$
C ~=~ \left\{ \bbB ~\Bigg|~ \forall i,~\sum_j b_{ij} ~=~ b_i~~\text{and}~~\forall i,j,~b_{ij}\ge 0 \right\}.
$
For our purpose, it suffices to know that any equilibrium $\bbB^*\in C$ of PRD corresponds to a minimum point of $g$.
The market potential with proportional response dynamics will be defined as:
\begin{equation}
\label{eqn:PR-dyn}
G(\bbB) ~=~ g(\bbB) - g(\bbB^*)
\end{equation}
Cheung, Cole and Tao~\cite{CCT2018} show that for positive constants $q_1 < q_2$ (which depend on the maximum and minimum values of $\rho_i$) the market potential in a static market is bounded by 
\[
G(\bbB^{t+1}) ~~\le~~ q_1 \cdot \KL(\bbB^*,\bbB^t) ~-~ q_2 \cdot \KL(\bbB^*,\bbB^{t+1}).
\]
%
%\pn{Our results are on the substitute domain. We will discuss the difficulty to generalize our technique to the complementary domain.}

In the rest of the section, we analyze the impact of changing utility functions and supplies on the convergence properties of proportional response dynamics. For the varying budgets case, the domain $C$ varies too, prohibiting a similar analysis.

\parabold{Dynamic Buyer Utilities.} Starting with the initial utility parameters, suppose that each $a_{ik}$ changes by a factor within $[e^{-\varepsilon},e^\varepsilon]$. For a given budget allocation $\bbB$, let $G(\calM^t,\bbB)$ denote the market potential for the utility of the buyers in round $t$, and $\bbB^{t,*} \in C$ the allocation that minimizes $G(\calM^t,\bbB)$.

\begin{proposition}
\label{prop:PRD}
After $T$ rounds, it holds that
\[
G(\calM^T,B^T) ~\leq~ q_1 \left( \frac{q_1}{q_2} \right)^{T-1} \cdot \KLe(\bbB^{0,*},\bbB^0) ~+~ \frac{q_2}{q_2 - q_1} \cdot \Delta\enspace,\]
where 
\[
\Delta ~=~ \sum_i \left(  \frac{b_i  (e^{\kappa_i} - 1)}{1 - \rho_i} \cdot \left| \rho_i \log \left( \frac{B}{ b_i } \right) - \log \left( \min\limits_{t,j}  \frac{a_{ij}^t}{\sum_k a_{ik}^t} \right) \right| + \frac{2 b_i \varepsilon}{\rho_i} \right)\enspace,
\]

and $\kappa_i = 2\varepsilon \left( 1 - c_i\left( 3 - 2 \min_i c_i \right) \right)$, where $c_i = \rho_i / (\rho_i - 1)$.
\end{proposition}

To prove this, we show in the following Claim~\ref{claim:PRD} that the potential~(\ref{eqn:PR-dyn}) satisfies a Market-Perturbation property with the value of $\Delta$ specified in the proposition.
\begin{claim}
For any round $t \le T$ it holds
\label{claim:PRD}
\begin{align*}
	G(\calM^{t+1},B^{t+1}) ~\leq~& q_1 \cdot \KLe(\bbB^{t,*},\bbB^t) ~-~ q_2 \cdot \KLe(\bbB^{t+1,*},\bbB^{t+1}) ~+~ 2 \sum_i \frac{b_i \varepsilon}{\rho_i}\\
	& ~+~ \sum_i b_i (e^{\kappa_i} - 1) \cdot \left| \log C_i - \log \Pi_i \right| ,
\end{align*} where $C_i =   \left( \frac{B}{ b_i} \right)^{\frac{\rho_i}{1 - \rho_i}} $ is a constant and $\Pi_i = \left( \min\limits_{t,j}  \frac{a_{ij}^t}{\sum_k a_{ik}^t} \right)^{\frac{1}{1 - \rho_i}} $.

\end{claim}
\begin{proof}[Proof Sketch.]
Consider the function $g$ defined above, and let $g_t(\bbB) = g(\calM^t,\bbB)$ be function $g$ in round $t$ with utility coefficients $\{a_{ij}^t\}$. We have
$$g_t(\bbB^{t+1}) - g_t(\bbB^{t,*}) ~~\le~~ q_1 \cdot \KL(\bbB^{t,*},\bbB^t) ~-~ q_2 \cdot \KL(\bbB^{t,*},\bbB^{t+1}).$$
To be able to prove this claim, we need to derive an inequality of the following form:
$$g_{t+1}(\bbB^{t+1}) - g_{t+1}(\bbB^{t+1,*}) ~~\le~~ q_1 \cdot \KL(\bbB^{t,*},\bbB^t) ~-~ q_2 \cdot \KL(\bbB^{t+1,*},\bbB^{t+1}) ~+~ \Delta_t.$$
This inequality is implied by the first one whenever $\Delta_t$ is chosen large enough to satisfy
\begin{align*}
\Delta_t ~\ge&~\left[ g_{t+1}(\bbB^{t+1}) - g_t(\bbB^{t+1}) \right] ~+~ \left[ g_t(\bbB^{t,*}) - g_{t+1}(\bbB^{t+1,*}) \right]\\
& \hspace{0.5cm} ~+~ q_2 \cdot \left[ \KL(\bbB^{t+1,*},\bbB^{t+1}) - \KL(\bbB^{t,*},\bbB^{t+1}) \right].
\end{align*}
We choose a value for $\Delta_t$ that satisfies the even larger lower bound of
\begin{align}
&\left[ g_{t+1}(\bbB^{t+1}) - g_t(\bbB^{t+1}) \right] ~+~ \left[ g_t(\bbB^{t+1,*}) - g_{t+1}(\bbB^{t+1,*}) \right] \nonumber\\
& \hspace{0.5cm} \label{eq:Deltat_UB} ~+~ q_2 \cdot \left[ \KL(\bbB^{t+1,*},\bbB^{t+1}) - \KL(\bbB^{t,*},\bbB^{t+1}) \right]\\
\le~& \nonumber 2\left(\sum_i \frac{b_i \varepsilon}{\rho_i} \right) ~+~ q_2 \cdot \left[ \KL(\bbB^{t+1,*},\bbB^{t+1}) - \KL(\bbB^{t,*},\bbB^{t+1}) \right] \\
\le~& \Delta_t \nonumber
\end{align}
The first inequality follows since for any $\bbB$, we have $g_{t+1}(\bbB) - g_t(\bbB) ~=~ -\sum_{ij} \frac{b_{ij}}{\rho_i} \log \frac{a_{ij}^{t+1}}{a_{ij}^{t}}$ and $\max_{j} |\log (a_{ij}^{t+1} / a_{ij}^{t})| = \varepsilon$. Finally, we derive an upper bound on the third term in Appendix~\ref{app:PRD}. This yields the final value of $\Delta_t$ and proves the claim.
\end{proof}

\begin{proof}[Proof of Proposition~\ref{prop:PRD}]
Now suppose $\Delta$ is given as in the proposition, then with the above claim it follows that:
\begin{equation}
\label{eqn:KL-bnd}
	q_2 \cdot \KL(\bbB^{t+1,*},\bbB^{t+1}) ~\leq~ q_1 \cdot \KL(\bbB^{t,*},\bbB^t) ~+~ \Delta.
\end{equation}
The potential of the market at round $T$ can be bounded by
\begin{align*}
G_{T}(B^{T}) & ~\leq~ q_1 \cdot \KL(\bbB^{T-1,*},\bbB^{T-1}) ~+~ \Delta\\[5pt]
	& ~\leq~ q_1 \left( \frac{q_1}{q_2} \cdot \KL(\bbB^{T-2,*},\bbB^{T-2}) + \frac{\Delta}{q_2} \right) ~+~ \Delta\\
	& ~\leq~ q_1 \cdot \left( \frac{q_1}{q_2} \right)^{T-1} \cdot \KL(\bbB^{0,*},\bbB^{0}) ~+~ \frac{q_2}{q_2 - q_1} \cdot \Delta,
\end{align*}
where the inequalities follow by recursive application of~\eqref{eqn:KL-bnd}.
\end{proof}

\parabold{Dynamic Supplies.}
It turns out that the case with varying supplies can be reduced to the case with varying utility functions. To see this, note that the function $g$ defined in \eqref{eq:PF-PRD} assumes that the supply of each good is normalized to be one unit. When the supply of good $j$ is changed from $1$ to $e^\epsilon$, by performing a re-normalization of the supply, it is equivalent to changing $a_{ij}$ to $a_{ij} \cdot e^{\epsilon \rho_i}$.

%obtaining a linear convergence rate for the resulting dynamic, assuming $0 < \rhoi < 1$.
%
%
%
%
%Instead, we observe that in the substitutes domain this rule is the mirror descent updating
%rule using a KL divergence for the following optimization problem.
%
%\begin{align*}
%& \min_{\bbb} \Phi(\bbb) = - \sum_{ij} \frac{\bij}{\rhoi} 
%\log \frac{\aij (\bij)^{\rhoi - 1} } 
%{\sum_{h}  b_{hj} } \\
%&
%\text{subject to}~~~~\textstyle{\sum_j} \bij \le \ei~~~~\text{for all } i,
%~~~~\text{and}~~\bij \geq 0~~~~\text{for all } i, j.
%\end{align*}
%%
%In the complementary domain, the mirror descent updating rule for this function is:

%\subparagraph*{Acknowledgements.}
%
%I want to thank \dots

%%
%% Bibliography
%%

%% Either use bibtex (recommended), 

\bibliographystyle{abbrvurl}% the recommended bibstyle
\bibliography{bib}

%% .. or use the thebibliography environment explicitely
\clearpage

\appendix

\section{Missing Proofs}
\label{app:Proofs}

\subsection{Proof of Proposition~\ref{prop:misspendingDynamicBudgets}}
\label{app:misspendingDynamicBudgets}

%\begin{proof}
To show the result, we establish the Market-Perturbation property. Note that the misspending potential can be given by
\begin{align*}
\Phims(\calM^t,\vecp^t) ~=~ \sum_{j=1}^n p_j^t \cdot \left| x_j^t - w_j \right|%\\
                         %&=~ \sum_{j=1}^n p_j^t \cdot \left| \sum_{i=1}^m x_{ij}^t - w_j \right| \\
                         &=~ \sum_{j=1}^n p_j^t \cdot \left| \sum_{i=1}^m \left(b_i^1+\sum_{\tau=1}^t \ep_i^\tau\right) \cdot \frac{(a_{ij})^{1-c} (p_j)^{c-1}}{\sum_{k=1}^n (a_{ik})^{1-c} (p_k)^c} - w_j \right|\enspace.
\end{align*}
Hence, by the triangle inequality,
\begin{align*}
  \Phims(\calM^t,\vecp) ~&\le~ \sum_{j=1}^n p_j \cdot \left| \sum_{i=1}^m \left(b_i^1+\sum_{\tau=1}^{t-1} \ep_i^\tau\right) \cdot \frac{(a_{ij})^{1-c} (p_j)^{c-1}}{\sum_{k=1}^n (a_{ik})^{1-c} (p_k)^c} - w_j \right|\\
  & ~~~~~~~+~~ \sum_{j=1}^n p_j \cdot \left| \sum_{i=1}^m \ep_i^t \cdot \frac{(a_{ij})^{1-c} (p_j)^{c-1}}{\sum_{k=1}^n (a_{ik})^{1-c} (p_k)^c} \right| \\
  &=~ \Phims(\calM^{t-1},\vecp) ~~+~~ \sum_{j=1}^n p_j \cdot \left| \sum_{i=1}^m \ep_i^t \cdot \frac{(a_{ij})^{1-c} (p_j)^{c-1}}{\sum_{k=1}^n (a_{ik})^{1-c} (p_k)^c} \right| \\
  &\le~ \Phims(\calM^{t-1},\vecp) ~~+~~ \sum_{i=1}^m |\ep_i^t| \cdot  \sum_{j=1}^n p_j \cdot \frac{(a_{ij})^{1-c} (p_j)^{c-1}}{\sum_{k=1}^n (a_{ik})^{1-c} (p_k)^c} \\
  &=~ \Phims(\calM^{t-1},\vecp) ~~+~~ \| \vept \|_1 
\end{align*}
Thus, using Proposition~\ref{prop:marketConverge} with $\Delta^t = \|\vept\|_1$ and $\Delta = \eta$, the proof follows.
%it holds for any $t \le T$
%$$
%\Phims(\calM^T,\vecp^T) ~\le~ \sum_{\tau=t+1}^{T} (1-\delta)^{T-\tau} \|\veptau\|_1 ~+~ \frac{(1-\delta)^{T-t}}{\delta}\cdot m\eta ~+~ (1-\delta)^T \cdot \Phims(\calM^0,\vecp^0)\enspace.\qedhere
%$$
\qed
%\end{proof}

\subsection{Proof of Lemma~\ref{lem:misspendingDynamicUtility}}
\label{app:misspendingDynamicUtility}

%\begin{proof}
We describe a simple algorithm to compute such a vector $\beta'$ and prove the lemma. Initially, set $\beta^0 = \beta$ and partition the set $G$ of goods into sets $S^0 = \left\{j \in G \mid \frac{\alpha_j\beta^0_j}{\sum_k \alpha_k \beta^0_k} \ge \alpha_j \right\}$ and $R^0 = G \setminus S$. This yields
\begin{align*}
 \sum_j \Delta s_j  
  ~&=~ \sum_j \mod{\frac{\alpha_j \beta^0_j}{\sum_{k} \alpha_k \beta^0_k } ~-~ \alpha_j }\\
   &=~ \sum_{j \in S^0} \frac{\alpha_j \beta^0_j}{\sum_k \alpha_k  \beta^0_k } - \alpha_j ~+~ \sum_{j \in R^0}  \alpha_j - \frac{\alpha_j \beta^0_j}{\sum_k \alpha_k \beta^0_k } \\
   &=~\frac{\sum_{j \in S^0} \alpha_j \beta^0_j - \sum_{j\in R^0} \alpha_j \beta^0_j}{\sum_k \alpha_k \beta^0_k} ~-~ \sum_{j \in S^0} \alpha_j + \sum_{j \in R^0} \alpha_j\enspace.
\end{align*}
For $j \in S^0$, the derivative of this expression for $\beta_j^0$ is $\frac{\alpha_j\left(2\sum_{k \in R^0}\alpha_k \beta^0_k \right)}{\left(\sum_k \alpha_k \beta^0_k\right)^2}$, whereas for $j \in R^0$ it is $- \frac{\alpha_j\left(2\sum_{k \in R^0}\alpha_k \beta^0_k \right)}{\left(\sum_k \alpha_k \beta^0_k\right)^2}$. Now, we call a good $j$ \emph{consistent} if $j \in S^0$ and $\beta_j^0 = \mu$, or $j \in R^0$ and $\beta_j^0 = 1/\mu$. If all goods are consistent, we are done. Otherwise, we can pick any inconsistent good $j$ and move $\beta_j^0$ to the exteme value (either to $\mu$ when $j \in S^0$ or $1/\mu$ when $j \in R^0$). We name the resulting vector $\beta^1$ and repeat the process: Construct $S^1$ and $R^1$, find any inconsistent good and move its value to the exteme value. We repeat the process until all goods are consistent. 

Suppose in some round $\ell-1$ we increase $j \in S^{\ell-1}$ to $\beta_j^\ell = \mu$. Then the expression $\frac{\alpha_j \beta_j^\ell}{\sum_k \alpha_k \beta^\ell_k}$ will increase. Hence, since $j \in S^{\ell-1}$ it must hold $j \in S^\ell$ and, thus,
\begin{equation}
  \label{eq:SstaysS}
    \frac{\mu}{\alpha_j \mu + \sum_{k} \alpha_k \beta_k^{\ell-1}} ~\ge~ 1 \enspace.
\end{equation}
Now consider any good $j'$ with $\beta_{j'}^{\ell-1} = \mu$. It must be $j' \in S^{\ell}$ since  using Eq.~\eqref{eq:SstaysS}
\[
    \frac{\alpha_{j'}\beta_{j'}^{\ell}}{\sum_{k} \alpha_k \beta_k^{\ell}} ~=~ 
    \frac{\alpha_{j'} \mu}{\alpha_j \mu + \sum_{k} \alpha_k \beta_k^{\ell-1}} ~\ge~ \alpha_{j'}\enspace.
\]
Now suppose in some round $\ell-1$ we decrease $j \in R^{\ell-1}$ to $\beta_j^\ell = 1/\mu$. Now consider any good $j' \in S^{\ell-1}$. Since the expression for $j'$ only increases, we have $j' \in S^{\ell}$. 

These observations show that if in some iteration $\ell'$ we have inconsistent $j' \in S^{\ell'}$ and raise the value to $\beta_{j'}^{\ell'+1} = \mu$, then $j' \in S^{\ell}$ and remains consistent for every $\ell \ge \ell'$. Using similar arguments, it follows that if in some iteration $\ell'$ we have inconsistent $j' \in R^{\ell'}$ and lower the value to $\beta_{j'}^{\ell'+1} = 1/\mu$, then $j' \in R^{\ell}$ and remains consistent for every $\ell \ge \ell'$. Thus, after at most $n$ iterations, the process is done and a consistent vector $\beta'$ has evolved. Furthermore, all adjustments are non-decreasing for the overall expression, and we obtain an upper bound on $\sum_j \Delta s_j$.
\qed

\subsection{Proof of Claim~\ref{claim:PRD}}
\label{app:PRD}

To finish the proof, it remains to bound third term in \eqref{eq:Deltat_UB}. The explicit expression for the third term is
\begin{equation}\label{eq:third-explicit}
\sum_{ij} \left( b_{ij}^{*,t+1} \cdot \log \frac{b_{ij}^{*,t+1}}{b_{ij}^{t+1}} - b_{ij}^{*,t} \cdot \log \frac{b_{ij}^{*,t}}{b_{ij}^{t+1}} \right).
\end{equation}
Recall that by setting $c_i = \rho_i / (\rho_i - 1)$, at $\bbp^{*,t}$, the total demand for each good $j$ is one, i.e.,
\begin{equation}\label{eq:demand-before}
\sum_i b_i \cdot \frac{(a^t_{ij})^{1-c_i} (p^{*,t}_j)^{c_i-1}}{\sum_{k=1}^n (a^t_{ik})^{1-c_i} (p^{*,t}_k)^{c_i}} ~=~ 1.
\end{equation}

Now, suppose that each $a_{ik}$ changes by a factor within $[e^{-\varepsilon},e^\varepsilon]$.
%After the change of values to $a_{ik}$'s,
Let $j = \argmax_k (p_k^{*,t+1} / p_k^{*,t})$, and let $\alpha = p_j^{*,t+1} / p_j^{*,t} \ge 1$.
In the substitute domain, the new total demand for good $j$ is
\begin{align*}
\sum_i b_i \cdot \frac{(a^{t+1}_{ij})^{1-c_i} (p^{*,t+1}_j)^{c_i-1}}{\sum_{k=1}^n (a^{t+1}_{ik})^{1-c_i} (p^{*,t+1}_k)^{c_i}}
&~\le~  \sum_i b_i \cdot \frac{\left[(a^{t}_{ij})^{1-c_i} \cdot e^{\varepsilon(1-c_i)}\right] \left[ (p^{*,t}_j)^{c_i-1} \cdot \alpha^{c_i-1}\right]}
{\sum_{k=1}^n \left[ (a^{t}_{ik})^{1-c_i} \cdot e^{-\varepsilon(1-c_i)} \right] \left[ (p^{*,t}_k)^{c_i} \cdot \alpha^{c_i} \right]}\\
&~=~ \sum_i b_i \cdot \frac{e^{2\varepsilon(1-c_i)}}{\alpha} \cdot \frac{\left[(a^{t}_{ij})^{1-c_i} \right] \left[ (p^{*,t}_j)^{c_i-1} \right]}
{\sum_{k=1}^n \left[ (a^{t}_{ik})^{1-c_i} \right] \left[ (p^{*,t}_k)^{c_i} \right]}\\
&~\le~ \frac{e^{2\varepsilon(1-\min_i c_i)}}{\alpha}.\comm{by \eqref{eq:demand-before}}
\end{align*}
But note that the new total demand is also $1$. Thus, $\alpha \le e^{2\varepsilon(1-\min_i c_i)}$,
and hence for every good $k$, $p_k^{*,t+1} / p_k^{*,t} \le e^{2\varepsilon(1-\min_i c_i)}$.

Symmetrically, we can also prove that for every good $k$, $p_k^{*,t+1} / p_k^{*,t} \ge e^{-2\varepsilon(1-\min_i c_i)}$.
To conclude, we show that
$$(p_k^{*,t+1} / p_k^{*,t}) \in [e^{-\delta},e^{\delta}],~\text{where}~\delta = 2\varepsilon(1-\min_i c_i).$$

Next, note that
$$
b_{ij}^{t,*} ~=~ b_i \cdot \frac{(a^t_{ij})^{1-c_i} (p^{*,t}_j)^{c_i}}{\sum_{k=1}^n (a^t_{ik})^{1-c_i} (p^{*,t}_k)^{c_i}}
$$
while
$$
b_{ij}^{t+1,*} ~=~ b_i \cdot \frac{(a^{t+1}_{ij})^{1-c_i} (p^{*,t+1}_j)^{c_i}}{\sum_{k=1}^n (a^{t+1}_{ik})^{1-c_i} (p^{*,t+1}_k)^{c_i}}.
$$
Since $(a^{t+1}_{ik} / a^{t}_{ik}) \in [e^{-\varepsilon},e^{\varepsilon}]$ and $(p_k^{*,t+1} / p_k^{*,t}) \in [e^{-\delta},e^{\delta}]$,
it is clear that
$$
\frac{b_{ij}^{t+1,*}}{b_{ij}^{t,*}} ~\in~ \left[e^{-2\left( \varepsilon(1-c_i) - c_i\delta \right)}~,~e^{2\left(\varepsilon(1-c_i) - c_i\delta\right)}\right]
~=:~ \left[e^{-\kappa_i}~,~e^{\kappa_i}\right]
$$

Now we come back to \eqref{eq:third-explicit}.
Note that $\frac{\partial}{\partial t} \left(t\log \frac{t}{q}\right) ~=~ \log \frac{t}{q} + 1$.
Since $\sum_{ij} b_{ij}^{*,t+1} = \sum_{ij} b_{ij}^{*,t}$, we can ignore the constant $+1$ in the above partial derivative,
and then bound $\eqref{eq:third-explicit}$ by
\begin{equation}
\label{eqn:temp-bnd}
\sum_{ij} \left|b_{ij}^{*,t+1} - b_{ij}^{*,t} \right| ~\cdot~
\left(~~\text{the maximum possible value of }\max_t\left\{\left| \log \frac{b_{ij}^{*,t'}}{b_{ij}^t} \right| \right\}~~\right).
\end{equation}
Without further assumption, the maximum possible value of $\max_{t,t'}\left\{\left| \log \frac{b_{ij}^{*,t'}}{b_{ij}^t} \right| \right\}$
can be arbitrarily big, since we might have a very tiny (yet positive) value of $b_{ij}^0$.
On the other hand, clearly, both $b_{ij}^t$ and $b_{ij}^{*,t'}$ is upper bounded by $b_i$.

So it suffices to derive a lower bound for $b_{ij}^t$ and $b_{ij}^{*,t'}$.
Since $b_{ij}^{*,t'}$ is an equilibrium value and it can be reached arbitrarily closely in the static setting,
a lower bound for $b_{ij}^t$ will carry through as a lower bound for $b_{ij}^{*,t'}$.

By \eqref{eqn:subst-CES-PR-rule}, if $b_{ij}^t \ge q$, then
$$b_{ij}^{t+1} ~\ge~ b_i \cdot \frac{a_{i,\min} \cdot (q/B)^{\rho_i}}{\sum_k a_{ik}}.$$
Thus, the value of $q_i$ which satisfies the equation $q_i ~=~ b_i\cdot \frac{a_{i,\min} \cdot (q_i/B)^{\rho_i}}{\sum_k a_{ik}}$
can serve as a lower bound of $b_{ij}^{t}$, provided that $b_{ij}^{0}\ge q_i$.
The equation solves to
$$
q_i ~=~ \left( b_i \cdot \frac{a_{i,\min}}{\sum_k a_{ik}} \cdot \frac{1}{B^{\rho_i}} \right)^{1/(1-\rho_i)}.
$$
Note that in the above bound, we deliberately ignore the fact that the values of $a_{ik}$ is changing, so as to avoid cluster of algebra.
We may simply replace $\frac{a_{i,\min}}{\sum_k a_{ik}}$ by $\min_{t,j} \frac{a_{ij}^t}{\sum_k a_{ik}^t}$ to complete the proof.

\noindent
Using this lower bound together with the fact that $\frac{b_{ij}^{t+1,*}}{b_{ij}^{t,*}} ~\in~ \left[e^{-\kappa}~,~e^{\kappa}\right]$ we can now bound~\eqref{eqn:temp-bnd} as:\[
	\sum_i b_i (e^{\kappa_i} - 1 ) \cdot \left| \log \left[ \frac{b_i B^{\rho_i / 1-\rho_i}}{\left( b_i \min_{t,j} \frac{a_{ij}^t}{\sum_k a_{ik}^t}  \right)^{1/1-\rho_i}} \right]  \right| ~=~ \sum_i b_i (e^{\kappa_i} - 1 ) \left| \log C_i - \log \Pi_i \right|,
\] where $C_i =   \left( \frac{B}{ b_i} \right)^{\frac{\rho_i}{1 - \rho_i}} $ is a constant and $\Pi_i = \left( \min\limits_{t,j}  \frac{a_{ij}^t}{\sum_k a_{ik}^t} \right)^{\frac{1}{1 - \rho_i}} $. This proves Claim~\ref{claim:PRD}. 
\qed
%\end{proof}

\section{Dynamic Fisher Markets via Convex Potential Functions}
\label{app:CPF}

%\subsection{Prior Results about Static CES Fisher Markets}
% !TEX root = main.tex

%\section{Dynamic Fisher Markets via Convex Potential Functions}\label{sec:CPF}

\newcommand{\Xicpf}{\Psi_{\textsf{CPF}}}
\newcommand{\Xicpfs}{\Xicpf^*}
\newcommand{\Xicpfst}{\Xicpf^{*,t}}
\newcommand{\Xicpfsto}{\Xicpf^{*,t+1}}

Suppose each buyer $i$ has a CES utility function
$u_i(\vecx_i) = \left( \sum_{j=1}^n a_{ij} \cdot (x_{ij})^{\rho} \right)^{1/\rho}$, where $1\ge \rho > -\infty$ and $a_{ij} \ge 0$.\footnote{For
simplicity, we assume all CES utility functions have the same $\rho$; but the convex potential function works even with distinct $\rho$'s after some obvious modifications.} Let $c =\rho/(\rho-1)$. Recall that $w_j$ is the supply of good $j$. The convex potential function for a static CES Fisher market is~\cite{CCD2013}
$$
\Xicpf(\calM, \vecp) ~=~ \sum_{j=1}^n w_j\cdot p_j  ~-~ \sum_i b_i \cdot \ln Q_i(\vecp),~~~~\text{where}~~Q_i(\vecp) ~=~ \left(\sum_{k=1}^n (a_{ik})^{1-c} (p_k)^c\right)^{1/c}.
$$
Note that $Q_i(\vecp)$ is independent of the supplies of goods and the budgets of buyers; it can be interpreted as the minimum amount of money buyer $i$ needs to use to earn one unit of utility~\cite{Cheung2014}. The minimum value of $\Xicpf(\calM, \vecp)$ is not zero in general. Hence, for applying our general framework that requires zero-value at the minimum, we use a normalized version $\Phicpf(\calM, \vecp) ~:=~ \Xicpf(\calM, \vecp) - \Xicpfs(\calM)$, where $\Xicpfs(\calM) := \min_{\vecp} \Xicpf(\calM, \vecp)$.

We study the following tatonnement price-update rule: 
\begin{equation}\label{eq:tat-update-CPF}
p_j^{t+1} ~\la~ p_j^t \cdot \left[1 + \lambda \cdot \min\{1~,~z_j^t\} \right],
\end{equation}
where $\lambda$ is a constant satisfying $0<\lambda<1/6$.

Let $\Xicpfs(\calM)$ denote the minimum value of the function $\Xicpf(\calM)$. The following theorem, stated in a simplified format from~\cite{CCD2013},
demonstrates the Price-Improvement property.
\begin{theorem}[~\cite{CCD2013}]
	Let $\vecp^0$ denote the initial prices and $\vecp^*$ denote the mark et equilibrium. Suppose prices are updated according to the rule \eqref{eq:tat-update-CPF}. If $\min_j p^0_j / p^*_j \ge q > 0$, then there exists $\delta = \delta(q,\lambda) > 0$ such that for any time $t\ge 0$, it holds $\Xicpf(\calM, \vecp^{t+1}) - \Xicpfs(\calM) ~\le~ (1-\delta)\cdot (\Xicpf(\calM, \vecp^t) - \Xicpfs(\calM))$.
\end{theorem}
For our dynamic environment, we denote the market at time $t$ by $\calM^t = (\vecu^t, \vecb^t, \vecw^t)$, and
$$
\Xicpf(\calM^t,\vecp^t) ~=~ \sum_{j=1}^n w_j^t \cdot p_j^t  ~-~ \sum_i b_i^t \cdot \ln Q_i^t(\vecp^t),~~~~\text{where}~~
Q_i^t(\vecp) ~=~ \left(\sum_{k=1}^n (a_{ik}^t)^{1-c} (p_k)^c\right)^{1/c}.
$$
Let $\Xicpfst ~=~ \min_{\vecp} \Xicpf(\calM^t,\vecp)$, and $\Phicpf(\calM^t, \vecp) = \Xicpf(\calM^t,\vecp) - \Xicpfst$.

Similar to our analysis with the misspending potential function, we establish the Market-Perturbation property in case supplies, budgets and utility functions are changing dynamically. 

%We only list the main results here and defer all details and proofs to Appendix~\ref{app:CPF}.
%
%For dynamic supply as in Section \ref{sect:dynamic-supply}, we obtain
%$$
%\Phicpf(\calM^{t+1}, \vecp^{t+1}) ~\le~ (1-\delta) \cdot \Phicpf(\calM^t, \vecp^t) ~+~ (P+B) \|\vept\|_1,~~~\text{i.e., }\Delta^t = (P+B) \|\vept\|_1.
%$$
%
%For dynamic budget, assuming that $b_i^t\in [b_i^0/C,C\cdot b_i^0]$ for some fixed constant $C\ge 1$, there exists a constant $C'$ which depends on $C$ and the initial prices only (see Appendix~\ref{app:CPF} for the precise relation), such that
%$$
%\Phicpf(\calM^{t+1}, \vecp^{t+1}) ~\le~ (1-\delta)\cdot \Phicpf(\calM^t, \vecp^t) ~+~ C'\cdot \sum_i \left|b_i^{t+1} - b_i^t\right|,
%~~~\text{i.e., }\Delta^t = C'\cdot \sum_i \left|b_i^{t+1} - b_i^t\right|.\vuppp
%$$
%
%For dynamic utilities, assume that each $a_{ij}$ can in each round be changed by some multiplicative factor $\chi_{ij}^t > 0$. Let $\chi^t = \max_{i,j} ((\chi_{ij}^t)^{-1/\rho},(\chi_{ij}^t)^{1/\rho})$ and $\chi = \max_t \chi^t$. Then
%$$
%\Phicpf(\calM^{t+1}, \vecp^{t+1}) ~\le~ (1-\delta)\cdot \Phicpf(\calM^t, \vecp^t)~+~ 2B\ln \chi^t,
%~~~\text{i.e., }\Delta^t = 2B\ln \chi^t.\vuppp\vuppp\vuppp\vuppp\vuppp
%$$

\subsection{Dynamic Supply}

Here, we consider the cases when the supplies are changing, while buyers' budgets and utility functions are fixed.
Thus, the function $Q_i^t$ and budget $b^t_i$ does not change over time, and we write $Q_i$ and $b_i$ instead.
\begin{align*}
&~~~~~~~\Phicpf(\calM^{t+1}, \vecp^{t+1}) ~=~ \Xicpf(\calM^{t+1},\vecp^{t+1}) - \Xicpfsto\\
&~\le~ \left(\sum_{j=1}^n w_j^t \cdot p_j^{t+1} ~-~ \sum_i b_i \cdot \ln Q_i(p^{t+1}) ~-~ \Xicpfst\right)
~+~ \sum_{j=1}^n p_j^{t+1} \cdot \left| w_j^{t+1} - w_j^t \right| ~+~ (\Xicpfst - \Xicpfsto)\\
&~=~ \left[\Xicpf(\calM^t,\vecp^{t+1}) - \Xicpfst\right]
~+~ \sum_{j=1}^n p_j^{t+1} \cdot \left| w_j^{t+1} - w_j^t \right| ~+~ (\Xicpfst - \Xicpfsto)\\
&~\le~ (1-\delta)\cdot \left[\Xicpf(\calM^t,\vecp^t) - \Xicpfst\right]
~+~ P \cdot \sum_{j=1}^n \left| w_j^{t+1} - w_j^t \right| ~+~ (\Xicpfst - \Xicpfsto)\\
&~=~ (1-\delta)\cdot \Phicpf(\calM^t,\vecp^t) ~+~ P \|\vept\| ~+~ (\Xicpfst - \Xicpfsto).
\end{align*}

\newcommand{\psto}{\vecp^{*,t+1}}
\newcommand{\pst}{\vecp^{*,t}}

Now we proceed to bounding $(\Xicpfst - \Xicpfsto)$.
Let $\psto$ denote the price vector which attains the minimum value of $\Xicpf(\calM^{t+1},\vecp)$. Then
\begin{align*}
\Xicpfsto
&~=~ \sum_{j=1}^n w_j^{t+1}\cdot \psto_j ~-~ \sum_i b_i \cdot \ln Q_i(\psto)\\
&~\ge~ \sum_{j=1}^n w_j^t \cdot \psto_j ~-~ \sum_i b_i \cdot \ln Q_i(\psto) ~-~ \sum_{j=1}^n \psto_j \cdot \left| w_j^{t+1} - w_j^t \right|\\
&~\ge~ \Xicpfst ~-~ \sum_{j=1}^n \psto_j \cdot \left| w_j^{t+1} - w_j^t \right|~~~~~~~~\text{(by definition of }\Xicpfst)\\
&~\ge~ \Xicpfst ~-~ B \|\vept\|_1.
\end{align*}
The last inequality holds since each equilibrium price is bounded above by $B$, the total amount of money in the market.
Thus, $(\Xicpfst - \Xicpfsto)$ is bounded above by $B\|\vept\|$.

Summarizing, we have showed
$$
\Phicpf(\calM^{t+1}, \vecp^{t+1}) ~\le~ (1-\delta) \cdot \Phicpf(\calM^t, \vecp^t) ~+~ (P+B) \|\vept\|_1,
$$
i.e., $\Delta^t = (P+B) \|\vept\|_1$.

\subsection{Dynamic Budgets}

Here, we consider the cases when the buyers' budgets are changing, while supplies and buyers' utility functions are fixed.
\begin{align*}
&~~~~~\Phicpf(\calM^{t+1}, \vecp^{t+1}) ~=~ \Xicpf(\calM^{t+1},\vecp^{t+1}) - \Xicpfsto\\
&~=~ \left( \sum_{j=1}^{n} w_j \cdot p_j^{t+1} ~-~ \sum_i b_i^t\cdot \ln Q_i(\vecp^{t+1}) ~-~ \Xicpfst \right)
~-~ \sum_i (b_i^{t+1} - b_i^t)\cdot \ln Q_i(\vecp^{t+1}) ~+~ (\Xicpfst - \Xicpfsto)\\
&~=~ \left[ \Xicpf(\calM^t,\vecp^{t+1}) - \Xicpfst\right]
~-~ \sum_i (b_i^{t+1} - b_i^t)\cdot \ln Q_i(\vecp^{t+1}) ~+~ (\Xicpfst - \Xicpfsto)\\
&~\le~ (1-\delta)\cdot \Phicpf(\calM^t, \vecp^t)
~-~ \sum_i (b_i^{t+1} - b_i^t)\cdot \ln Q_i(\vecp^{t+1}) ~+~ (\Xicpfst - \Xicpfsto).
\end{align*}
\begin{align*}
\Xicpfsto
&~=~ \sum_{j=1}^{n} w_j\cdot \psto_j ~-~ \sum_i b_i^{t+1}\cdot \ln Q_i(\psto)\\
&~=~ \sum_{j=1}^{n} w_j\cdot \psto_j ~-~ \sum_i b_i^t \cdot \ln Q_i(\psto) ~-~ \sum_i (b_i^{t+1} - b_i^t)\cdot \ln Q_i(\psto)\\
&~\ge~ \Xicpfst ~-~ \sum_i (b_i^{t+1} - b_i^t)\cdot \ln Q_i(\psto).
\end{align*}

Combining the above two inequalities yields
$$
\Phicpf(\calM^{t+1}, \vecp^{t+1}) ~\le~ (1-\delta)\cdot \Phicpf(\calM^t, \vecp^t)
~+~ \sum_i (b_i^{t+1} - b_i^t) \cdot \ln \frac{Q_i(\psto)}{Q_i(\vecp^{t+1})}.
$$
Cheung et al.~\cite[Section 6.3]{CCD2013} showed that in the static market setting,
if the initial prices are neither too high nor too low, then $\frac{Q_i(\psto)}{Q_i(\vecp^{t+1})}$
has time-independent upper and lower bounds.
In the dynamic market setting, we assume that there exists a constant $C\ge 1$ such that
the budget of each buyer $i$ changes within the range $[b_i^0/C,C\cdot b_i^0]$.
Let $U^*,L^*$ be the time-independent upper and lower bounds derived in~\cite{CCD2013},
for the static market setting with $\vecb = (b_1^0,\ldots,b_m^0)$.
Following the argument in~\cite{CCD2013}, their upper bound on $p_k^{t+1}$ can be carried through to the dynamic market setting
by increasing by a factor of $C$,
while their lower bound on $p_k^{t+1}$ can be carried through to the dynamic market setting by shrinking by a factor of $1/C$;
these hold similarly for the equilibrium prices.
Thus, for the dynamic market setting, we have time-independent upper and lower bounds on $\frac{Q_i(\psto)}{Q_i(\vecp^{t+1})}$
of values $C^2\cdot U^*$ and $L^*/C^2$ respectively.
Thus, by setting
$$
C' ~:=~ \max \left\{~\left|\ln (C^2\cdot U^*)\right| ~,~ \left|\ln (L^*/C^2)\right|~\right\},
$$
we have
$$
\Phicpf(\calM^{t+1}, \vecp^{t+1}) ~\le~ (1-\delta)\cdot \Phicpf(\calM^t, \vecp^t) ~+~ C'\cdot \sum_i \left|b_i^{t+1} - b_i^t\right|,
$$
i.e., $\Delta^t = C'\cdot \sum_i \left|b_i^{t+1} - b_i^t\right|$.

\subsection{Dynamic Buyer Utility}

Here, we consider the cases when the buyers' utility function are changing, while supplies and budgets are fixed.
The changes to utility functions induce changes to the functions $Q_i^t$.

\begin{align*}
&~~~~~\Phicpf(\calM^{t+1}, \vecp^{t+1}) ~=~ \Xicpf(\calM^{t+1},\vecp^{t+1}) - \Xicpfsto\\
&~=~ \left( \sum_{j=1}^{n} w_j \cdot p_j^{t+1} ~-~ \sum_i b_i\cdot \ln Q_i^t(\vecp^{t+1}) ~-~ \Xicpfst \right)
~-~ \sum_i b_i\cdot \ln \frac{Q_i^{t+1}(\vecp^{t+1})}{Q_i^t(\vecp^{t+1})} ~+~ (\Xicpfst - \Xicpfsto)\\
&~=~ \left[ \Xicpf(\calM^t,\vecp^{t+1}) - \Xicpfst\right]
~-~ \sum_i b_i\cdot \ln \frac{Q_i^{t+1}(\vecp^{t+1})}{Q_i^t(\vecp^{t+1})} ~+~ (\Xicpfst - \Xicpfsto)\\
&~\le~ (1-\delta)\cdot \Phicpf(\calM^t, \vecp^t) ~-~ \sum_i b_i\cdot \ln \frac{Q_i^{t+1}(\vecp^{t+1})}{Q_i^t(\vecp^{t+1})} ~+~ (\Xicpfst - \Xicpfsto).
\end{align*}

\begin{align*}
\Xicpfsto
&~=~ \sum_{j=1}^{n} w_j\cdot \psto_j ~-~ \sum_i b_i\cdot \ln Q_i^{t+1}(\psto)\\
&~=~ \sum_{j=1}^{n} w_j\cdot \psto_j ~-~ \sum_i b_i \cdot \ln Q_i^t(\psto) ~-~ \sum_i b_i\cdot \ln \frac{Q_i^{t+1}(\psto)}{Q_i^t(\psto)}\\
&~\ge~ \Xicpfst ~-~ \sum_i b_i\cdot \ln \frac{Q_i^{t+1}(\psto)}{Q_i^t(\psto)}.
\end{align*}

Combining yields
$$
\Phicpf(\calM^{t+1}, \vecp^{t+1}) ~\le~ (1-\delta)\cdot \Phicpf(\calM^t, \vecp^t)
~+~ \sum_i b_i \cdot \ln \left(\frac{Q_i^{t+1}(\psto)}{Q_i^t(\psto)} \cdot \frac{Q_i^t(\vecp^{t+1})}{Q_i^{t+1}(\vecp^{t+1})} \right).
$$
Starting from the initial utility values, each $a_{ij}$ can in each round be changed by some multiplicative factor $\chi_{ij}^t$.
Let $\chi^t = \max_{i,j} ((\chi_{ij}^t)^{-1/\rho},(\chi_{ij}^t)^{1/\rho})$ and $\chi = \max_t \chi^t$.
Note that $(1-c)/c = -1/\rho$, so $1/\chi^t \le Q_i^{t+1}(\vecp) / Q_i^t(\vecp) \le \chi^t$ for any price vector $\vecp$.
Thus,
$$
\left|\ln \left(\frac{Q_i^{t+1}(\psto)}{Q_i^t(\psto)} \cdot \frac{Q_i^t(\vecp^{t+1})}{Q_i^{t+1}(\vecp^{t+1})} \right)\right| ~\le~ 2\ln \chi^t,
$$
and hence
$$
\Phicpf(\calM^{t+1}, \vecp^{t+1}) ~\le~ (1-\delta)\cdot \Phicpf(\calM^t, \vecp^t)~+~ 2B\ln \chi^t,
$$
i.e., $\Delta^t = 2B\ln \chi^t$.

\section{Further Applications}\label{sect:applications}

\subsection{Gradient Descent on Shifting Strongly Convex Functions }

In this section, we analyze the performance of gradient descent on a sequence of convex functions satisfying $\alpha$-strong convexity and $\beta$ smoothness. Specifically, let $\calF(\alpha, \beta)$ be a family of convex functions where each function is $\beta$-smooth and $\alpha$-strongly convex. Any individual function in this family is parameterized by the minimizing point of the function. It is known that gradient descent algorithm converges geometrically for smooth and strongly convex functions. Specifically, if $x^*$ is the minimizing point of the function and $x^t$ is the point chosen by gradient descent update in round $t$, then:
\begin{theorem}[Theorem 3.12,~\cite{Bubeck}]
Let $f$ be a $\beta$-smooth and $\alpha$-strongly convex function. Then gradient descent algorithm with step size $\eta_t \leq \frac{2}{\alpha + \beta}$ on $f$ satisfies
\[
	\norm{x^{t+1} - x^*}^2 ~\leq~ \left( 1 - \frac{2 \eta_t \alpha \beta}{\alpha + \beta} \right) \norm{x^{t} - x^*}^2.
\]
\end{theorem}
Since the gradient descent dynamic converges rapidly in the static case, one would expect the dynamic to also closely follow the minimizing point when the underlying function shifts only a little bit every round. We formalize this notion next.

We note that the system can be viewed as an instance of the Lyapunov dynamical system described in Section~\ref{sect:general-framework} where the gradient descent update is the evolution rule, the optimum point $x^{*,t}$ the system parameter in round $t$ and the distance $G_{x^{*,t}}(x) = \norm{x - x^{*,t}}$ the Lyapunov function. If we denote by $\Delta^t =  \norm{\xs{t+1} - \xs{t}}$, then $ G_{x^{*,t+1}} \leq G_{x^{*,t}} + \Delta^t$. We obtain the following proposition.

\begin{proposition}
\label{prop:gradient-descent}
Let $\calF(\alpha, \beta)$ be a family of $\alpha$-strongly convex and $\beta$-smooth functions and let $f^t  \in \calF $ be the function chosen by the system in round $t$.
\begin{enumerate}
\item For $\Phi(x^{*,t}, x^t) = G_{x^{*,t}}(x^t)$, $\Delta^t := \norm{\xs{t+1} - \xs{t}}$ and $\delta = \frac{2 \eta_t \alpha \beta}{\alpha + \beta}$, we have
$$ \Phi(x^{*,T}, x^T)  ~\leq~  (1 - \delta)^{T/2} \cdot \Phi(x^{*,0}, x^0) ~+~ \sum_{t=1}^T (1-\delta)^{\frac{T-t}{2}} \cdot \Delta^t\enspace .$$
\item Further, if $\Delta^t \leq d$ for all rounds $t$, then: $$\Phi(x^{*,T}, x^T) ~\leq~  (1 - \delta)^{T/2} \cdot  \Phi(x^{*,0}, x^0) ~+~ \frac{2d}{\delta}\enspace. $$
\end{enumerate}
\end{proposition}

\begin{proof}
The first part follows from Theorem~\ref{thm:meta} and the fact that $\norm{x^{t+1} - \xs{t}} ~\leq~ \left( 1 - \delta \right)^{1/2} \norm{x^{t} - \xs{t}}$. The second part follows by taking the sum of the geometric series.
%\begin{align*}
%\Phi(T)  ~\leq~ &  (1 - \delta)^{T/2} \Phi(0) ~+~ d \sum\limits_{t=0}^{T-1} (1 - \delta)^{t/2} \\[4pt]
%	~\leq~ &  (1 - \delta)^{T/2} \Phi(0) ~+~ \frac{d}{1 - (1 - \delta)^{1/2}} \\[4pt]
%	~\leq~ &  (1 - \delta)^{T/2} \Phi(0) ~+~ \frac{2d}{\delta}.\qedhere
%\end{align*}
\end{proof}
The proposition implies that irrespective of the starting position, the gradient descent dynamic \textit{follows} the optima such that the chosen sequence of points are always approximately within $2d/\delta$ of the optimal.

\parabold{A Remark on Regret.}
The bound in Proposition~\ref{prop:gradient-descent}(2) can be extended to obtain corresponding guarantees in terms of regret.  
\begin{align*}
\sum\limits_{t}  f^t(x^t) - f^t(\xs{t})  ~\leq~ &  \sum\limits_{t}  \frac{\beta}{2} \norm{x^{t} - \xs{t}}^2  ~\leq~  \sum\limits_{t}  \frac{\beta}{2} \left( (1 - \delta)^{t/2}  \Phi(x^{*,0}, x^0) ~+~ \frac{2d}{\delta}\right)^2  \\[4pt]
~\leq~ & \beta \sum\limits_{t} \left(  \left( (1 - \delta)^{t/2}  \Phi(x^{*,0}, x^0) \right)^2  + \left(\frac{2d}{\delta}\right)^2  \right)\\[4pt]
~\leq~ & \frac{\beta  d}{\delta}  \left( \Phi(x^{*,0}, x^0)\right)^2  + \beta T \left(\frac{2d}{\delta}\right)^2 .
\end{align*}

The first inequality follows by $\beta$-smoothness of the functions. This bound implies that the average regret incurred as the dynamic approaches the region of radius $2d/\delta$ from the optimal point is bounded by a constant. 

%Similar results have already been proved before in context of dynamic regret bounds. Note that since the optimal point shifts by $O(T)$ in $T$ rounds, obtaining a sub-linear regret is not possible. Nevertheless, we can interpret this bound as the total loss incurred as the update process approaches the neighborhood of the optimal points.

\subsection{Load Balancing with Dynamic Machine Speed}

Consider a setting with $n$ distinct machines all connected to each other to form an arbitrary network. For ease of notation, we label the machines as $m_i$ for $i=1$ to $n$. Each machine $m_i$ can process jobs at speed $s_i$. Jobs/tasks, assumed to be infinitely divisible, of total weight $M$ are arbitrarily distributed over the network. Our goal is to design a decentralized load balancing algorithm with the objective that the total processing time over all machines is minimized. 

\begin{wrapfigure}{L}{0.5\textwidth}
    \begin{minipage}{0.5\textwidth}
      \begin{algorithm}[H]
      \DontPrintSemicolon
		\caption{Diffusion}
		\label{algo:diffuse}
		  		  
		\For{$t=1$ to $T$}{
		  	\For{each machine $m_i$}{
		  			$f_i^{(t)} \leftarrow \text{total processing time on } m_i$ \;
		  			broadcast $f_i^{(t)}$ to all $j \in nbd(m_i)$ \;
		  			\ForAll{$j \in nbd(m_i)$}{
		  				\textbf{if} $f_i^{(t)} > f_j^{(t)}$ \textbf{then} send $P_{ij} (f_i^{(t)} - f_j^{(t)}) s_i$ load to $j$.\;
		  				
		  			}
				}
		  }
		\end{algorithm}
    \end{minipage}
  \end{wrapfigure}

Before proceeding, we set up some notation. $\vecs$ denotes the vector of machine speeds. $\vecll^{(t)} = (\ell_i^{(t)})_{i}$ denotes the vector of loads and $ \vecf^{(t)} = ( f_i^{(t)})_{i}$ the corresponding finishing times at round $t$. We assume throughout that the total load stays constant i.e. $\sum_i \ell_i^{(t)} = M$. For machine speed $
\vecs$, $f^{*, \vecs}$ denotes the corresponding vector of finishing times in the balanced state, i.e. a state where the finishing time of all machines is the same.
%\YKC{It looks to me that $f_i$ is defined as $\ell_i / s_i$. Thus, it's better to define $\ell_i$ and $s_i$ first, then $f_i$. Paresh, please change accordingly.}

Algorithm~\ref{algo:diffuse} is based on the diffusion principle~\cite{sub94}, where if a machine has more jobs than its neighbours, then some jobs diffuse to the neighbour. In our context, since the goal is to equalize the finishing times of all machines, the number of jobs that diffuse is proportional to the difference in the finishing times. The proportionality constant depends on the connecting edge. Specifically, in the algorithm that follows we use a diffusivity matrix $P$ satisfying the following conditions: (a) $P_{ii} \geq 1/2$ (b) $P_{ij} > 0$ iff $(i,j)$ is an edge in $G$. (c) $P$ is symmetric and stochastic, i.e., for every machine $m_i$, $\sum_j P_{ij} = 1$.

%\begin{algorithm}
%\caption{DIFFUSE}
%\label{algo:diffuse}
%  
%\begin{algorithmic}[1]
%  
%  \For{$t= 1$ to $T$}
%  		\For{\text{each machine $m_i$}}
%  			\State $f_i^{(t)} \leftarrow \text{total processing time on $m_i$}$
%  			\State broadcast $f_i^{(t)}$ to all $j \in nbd(m_i)$
%  			\ForAll{$j \in nbd(m_i)$}
%  				\If{$f_i^{(t)} > f_j^{(t)} $}
%  					\State Send $P_{ij} \cdot (f_i^{(t)} - f_j^{(t)}) \cdot s_i$ load to $j$.
%  				\EndIf
%  			\EndFor
%		\EndFor
%  \EndFor
%  \end{algorithmic}
%\end{algorithm}

If each machine $m_i$ uses the load balancing protocol as described above, then the finishing time of machine $m_i$ at time $t+1$ is:
\begin{align*}
f_i^{(t+1)} ~=~  & \frac{\ell_i^{(t+1)}}{s_i} ~=~ \frac{1}{s_i} \left( \ell_i^{(t)} - \sum\limits_{j=1}^n P_{ij} \left( \frac{\ell_i^{(t)}}{s_i} - \frac{\ell_j^{(t)}}{s_j} \right) s_i \right)\\
~=~  & \frac{1}{s_i} \left( \ell_i^{(t)} - \sum\limits_{j=1}^n P_{ij} \ell_i^{(t)} + \sum\limits_{j=1}^n P_{ij} \frac{\ell_j^{(t)}}{s_j} s_i \right)\\
~=~  & \frac{1}{s_i} \left( \sum\limits_{j=1}^n P_{ij} f_j^{(t)} s_i \right) = \left( P f^{(t)} \right)_i
\end{align*}
It therefore follows that $ \vecf^{(t+1)} =  P \vecf^{(t)} $. Further, in the balanced state $\vecf^*$, since the finishing time of all machines is the same,
\[
	\left( P \vecf^{*} \right)_i ~=~ \sum\limits_{j=1}^n P_{ij} f_j^{*} =  f_i^{*} \sum\limits_{j=1}^n P_{ij} = f_i^{*}
\] i.e. $P \vecf^{*}  = \vecf^{*}$. If we denote the \textit{error} in round $t+1$ by $\vece^{(t+1)}$, then:
\[
	\vece^{(t+1)} ~=~  \vecf^{(t+1)} - \vecf^*  ~=~  P (\vecf^{(t)} - \vecf^*)  ~=~ P \vece^{(t)}
\] i.e. the same transformations apply to the error vector as well. Since $P$ is a symmetric matrix, it has $n$ eigenvalues $ \lambda_1 , \lambda_2 \cdots \lambda_n$ and linearly independent corresponding eigenvectors. By theory of Markov chains, it is also known that $1 = \mod{\lambda_1} \geq   \mod{\lambda_2} \geq \cdots \mod{\lambda_n}$. Since $P$ scales the length of $\vece^{(t)}$ by a factor of $ \leq \mod{\lambda_2}$:
\begin{equation}
\label{eqn:err-bound}
	\norm{\vece^{(t+1)}} ~=~  \norm{P \vece^{(t)}} ~\leq~ \mod{\lambda_2} \norm{\vece^{(t)}} ~\Rightarrow~ \norm{\vece^{(t+1)}}  ~\leq~  \mod{\lambda_2}^t \norm{\vece^{(0)}}.
\end{equation}
\noindent
For a given speed vector $\vecs$, one can define the ``potential" as the normed distance: $\norm{\vecf^{(t)} - \vecf^{*, \vecs}}_1$. This measures the imbalance in the network in terms of the finishing times. From \eqref{eqn:err-bound}, since the error vector $\vece$ converges to zero linearly, the potential at balanced state is zero. Note that this load balancing setting is identical to the Lyapunov dynamical system introduced in Section~\ref{sect:general-framework}. Specifically, the speed vector $\vecs$ is the \textit{system parameter}, the evolution function $F(\vecll^{(t)})$ is the diffusion process as described in Algorithm~\ref{algo:diffuse} and the potential as mentioned above corresponds to the Lyapunov function $G_{\vecs}(\vecll^{(t)}) = G_{\vecs}^t$. Note that by \eqref{eqn:err-bound} it follows that $G_{\vecs}^{t+1} ~\leq~  \mod{\lambda_2}^t G_{\vecs}^{0}$. In the following, all norms are assumed to be L1 norms.

\begin{proposition}
For a speed vector $\vecs$ and an arbitrary load profile vector $\vecll$, let $\vecf$ denote the corresponding finishing time vector. For a Lyapunov function defined as $G_{\vecs} = \norm{\vecf - \vecf^{*, \vecs}}$, if the speed vector changes to $\vecs'$ for the same load profile, then:
\[
	G_{\vecs'} ~\leq~ G_{\vecs} ~+~ Mn \mod{\frac{1}{\norm{\vecs'}} - \frac{1}{\norm{\vecs}} }.
\]
\end{proposition}
\begin{proof}

For speed vector changes $\vecs'$ and the same load profile, the Lyapunov function is given by:
\begin{equation}
\label{eqn:Lyanpunv-fn}
	G_{\vecs'} ~=~ \norm{\vecf - \vecf^{*, \vecs'}} ~\leq~ \norm{\vecf - \vecf^{*, \vecs}} ~+~ \norm{\vecf^{*, \vecs'} - \vecf^{*, \vecs}} ~=~ G_{\vecs} + \norm{\vecf^{*, \vecs'} - \vecf^{*, \vecs}}.
\end{equation}

Let $\ell_i$ denote the load on machine $m_i$. The optimal load on the machines in a balanced state can be characterized using the following optimization problem:
\begin{align*}
	\min~  \sum\limits_i \frac{\ell_i}{s_i}  \hspace{1cm} \text{s.t}~~ \sum\limits_i  \ell_i = M.
\end{align*}
Using the underlying symmetry, we can claim that the load on any machine $m_i$ in the balanced state and its corresponding finishing time are $ \ell_i^* = \frac{s_i \cdot M}{\sum_k s_k}$ and $f_i^{*, \vecs} = \frac{\ell_i^*}{s_i} = \frac{M}{\sum_k s_k}$ respectively. It then follows that:
$$
~~~~~~~~~~~\norm{\vecf^{*, \vecs'} - \vecf^{*, \vecs}} ~=~ \sum_i \mod{ \frac{\ell^{'*}_i}{s'_i} -  \frac{\ell_i^*}{s_i}} ~=~ \sum_i \mod{ \frac{M}{\sum_k s'_k} - \frac{M}{\sum_k s_k} } ~=~ Mn \mod{\frac{1}{\norm{\vecs'}} - \frac{1}{\norm{\vecs}} } .~~~~~~~~~~~\qedhere
$$
\end{proof}

To formalize the problem, let $\calL\calB(N, M)$ be a family of load balancing environments where $N$ denotes the network of underlying machines and $M$ the total weight of jobs. Each individual environment $LB_{\vecs} \in \calL\calB(N, M)$ is parameterized by the machine-speed vector $\vecs$. The corresponding potential (also Lyapunov) function is denoted by $G_{\vecs}$.

\begin{proposition}
\label{thm:load-balancing}
Let $\calL\calB(N, M)$ be a family of load balancing environments on $n$ machines with the corresponding diffusivity matrix being $P_N$. Let $\vecs^0, \vecs^1, \cdots , \vecs^T$ denote the vector of machine speeds at times $0, 1 \cdots T$ respectively. If we denote by $\lambda_2$ the second largest eigenvalue of $P_N$ and $\Phi(\vecs^t, \vecll^t) := G_{\vecs^t}(\vecll^t)$, then
\[
	\Phi(\vecs^T,\vecll^T) ~\le~ \mod{\lambda_2}^T \cdot \Phi(\vecs^0,\vecll^0) ~+~ Mn  \sum_{t=1}^T \mod{\lambda_2}^{T-t} \cdot \mod{\frac{1}{\norm{\vecs^{t+1}}} - \frac{1}{\norm{\vecs^{t}}} }\enspace .
\]
\end{proposition}

\begin{proof}
The result follows from the fact that $G_{\vecs}^{t+1} \leq \mod{\lambda_2} G_{\vecs}^{t}$ and Theorem~\ref{thm:meta}.
\end{proof}
Since $\Phi$ is a measure of load imbalance in the network in terms of finishing times, the above theorem implies that if the change in the speed vectors across rounds is small, then the imbalance at time $T$ is small and depends largely on the most recent changes.
% !TEX root = main.tex

\section{Generalization to PLCLDS using Bregman Divergence}
\label{app:Bregman}

Recently, in the context of mirror descent on convex functions,
Cheung, Cole and Tao~\cite{CCT2018} proposed a general framework for demonstrating linear convergence
when the underlying function is \emph{strongly Bregman convex}, a new generalization of the standard notion of strong convexity.
% \pn{What is this exactly?}
They use it to analyze \emph{proportional response dynamic} (PRD) in static Fisher markets.
We propose a variant of PLCLDS based on their framework that uses the generalization of Bregman divergences as distance measure.

Let $C$ be a compact and convex set. Given a differentiable convex function $h(\bbx)$ with domain of $C$,
the \emph{Bregman divergence} generated by the \emph{kernel} $h$ is denoted by $d_h$, defined as
$$
d_h(\bbx,\bby) ~=~ h(\bbx) ~-~ \left[~h(\bby) ~+~ \inner{\nabla h(\bby)}{\bbx-\bby}~\right],~~~~\forall \bbx\in C~\text{and}~\bby~\in \rint(C),
$$
where $\rint(C)$ is the relative interior of $C$. % We note that in general, $d_h$ is asymmetric, i.e., $d_h(\bbx,\bby) \neq d_h(\bby,\bbx)$.
A convex function $f$ is $(\sigma,L)$-strongly Bregman convex w.r.t.~Bregman divergence $d_h$ if, $0<\sigma\le L$, and for any $\bby\in \rint(C)$
and $\bbx\in C$,
$$
f(\bby) + \inner{\nabla f(\bby)}{\bbx-\bby} + \sigma\cdot d_h(\bbx,\bby)~~\le~~ f(\bbx) ~~\le~~ f(\bby) + \inner{\nabla f(\bby)}{\bbx-\bby} + L\cdot d_h(\bbx,\bby).
$$

Note that the KL divergence used in the analysis of proportional response dynamics
in Section~\ref{sect:damped-PR} is an instance of Bregman divergence where the kernel function is $h(\bbx) = \sum_j (x_j \cdot \ln x_j - x_j)$.
If the system is static, the variant of PLCLDS satisfies properties (a) and (b), with property (c) replaced by the following new property:
there exists positive numbers $q_1<q_2$
%a \emph{decay parameter} $\delta > 0$
such that for any $\vecp,\vecps\in C$,
\begin{equation}\label{eq:strong-Bregman}
G(F(\vecp))
~\le~ q_1 \cdot d_h(\vecps,\vecp) ~-~ q_2 \cdot d_h(\vecps,F(\vecp)).
\end{equation}
The above property holds when, for instance, $G$ is a $(\sigma,L)$-strongly Bregman convex function with minimum value zero, % \pn{not defined}
and $F$ is a mirror descent update:
$$
F(\vecp) ~=~ \argmin_{\vecp '} ~\left\{~\inner{\nabla G(\vecp)}{\vecp ' - \vecp} + L\cdot d_h(\vecp ',\vecp)~\right\},
$$
for which $q_1 = L-\sigma$ and $q_2 = L$.~\cite{CCT2018}

By a suitable telescoping with \eqref{eq:strong-Bregman}, it is easy to show that $G(\vecp^T)~\le~ q_1 \cdot \left( q_1 / q_2 \right)^{T-1} \cdot d_h(\vecps,\vecp^0)$,
where $\vecps$ is \emph{any} fixed point (equilibrium) of the Lyapunov system.
However, for system which is dynamic, we will need a modification of the above property, presented in the theorem below.

\begin{theorem}
\label{thm:PLCLDS-variant}
Let $\calL$ be a PLCLDS with $\delta\equiv \delta(\calL) > 0$.
Let $\vecs^0,\vecs^1,\ldots,\vecs^T$  and $\bbp^{*,0},\bbp^{*,1},\ldots,\bbp^{*,T}$ denote the sequence of system parameters and fixed points
at times $0,1,\cdots,T$, respectively and let $\Phi(\vecs^t,\vecp^t) = G_{\vecs^t}(\vecp^t)$.
Suppose that for every $t=1,\ldots,T$ the system parameters
$\vecs^{t-1},\vecs^t \in\rr^d$ invoke a change such that the fixed points change from $\bbp^{*,t-1}$ to $\bbp^{*,t}$,
and:
$$
\Phi(\vecs^t,\vecp^t) ~\le~ q_1 \cdot d_h(\bbp^{*,t-1},\vecp^{t-1}) ~-~ q_2 \cdot d_h(\bbp^{*,t},\vecp^{t}) ~+~ \Delta^t ,
$$
If the initial control variable vector is denoted by $\vecp^0$, and the system evolves such that for every $t\ge 1$ we have $\vecp^t = F_{\vecs_{t-1}}(\vecp^{t-1})$, then $$
\Phi(\vecs^T,\vecp^T)
~~\le~~ q_1 \cdot \left( \frac{q_1}{q_2} \right)^{T-1} d_h(\bbp^{*,0},\vecp^{0}) ~+~ \sum\limits_{i=0}^{T-1} \left( \frac{q_1}{q_2} \right)^i \Delta^{T-i}.
$$
\end{theorem}

\begin{proof}

For $t\ge 1$, the system is assumed to satisfy the perturbation property,
\begin{equation}
\label{eqn:bregman-perturbation}
\Phi(\vecs^t,\vecp^t) ~\le~ q_1 \cdot d_h(\bbp^{*,t-1},\vecp^{t-1}) ~-~ q_2 \cdot d_h(\bbp^{*,t},\vecp^{t}) ~+~ \Delta^t .
\end{equation}

Note that if $\vecs^{t-1} = \vecs^t$, the above assumption reduces to \eqref{eq:strong-Bregman} since we can have $\bbp^{*,t-1} = \bbp^{*,t}$ and $\Delta^t=0$. Under this assumption, we can bound the Bregman divergence between the control variable $\vecp^t$ and any fixed point $\bbp^{*,t}$ in round $t$ as\[
		d_h(\bbp^{*,t}, \vecp^t) ~\leq~ \frac{q_1}{q_2} d_h(\bbp^{*,t-1}, \vecp^{t-1}) ~+~ \frac{\Delta_t}{q_2}.
\] Since each term $\Phi(\vecs^t,\vecp^t)$ and $d_h(\bbp^{*,t},\vecp^{t})$ is non-negative, the potential at round $T$ by property~(\ref{eqn:bregman-perturbation}) is:
\begin{align*}
\Phi(\vecs^T,\bbp^{*,T}) & ~\leq~ q_1 \cdot d_h(\bbp^{*,T-1},\vecp^{T-1}) ~+~ \Delta^T\\[5pt]
	& ~\leq~ q_1 \left( \frac{q_1}{q_2} d_h(\bbp^{*,T-2},\vecp^{T-2}) + \frac{\Delta^{T-1}}{q_2} \right) ~+~ \Delta^T\\
	& \vdots \\
	& ~\leq~ q_1 \cdot \left( \frac{q_1}{q_2} \right)^{T-1} d_h(\bbp^{*,0},\vecp^{0}) ~+~ \sum\limits_{i=0}^{T-1} \left( \frac{q_1}{q_2} \right)^i \Delta^{T-i}.
\end{align*} where the inequalities follow by recursive application of~\eqref{eqn:bregman-perturbation}.
\end{proof}

\end{document}